\documentclass[12pt]{article}
\usepackage{amsmath}
\usepackage{graphicx}
\usepackage{enumerate}
\usepackage[authoryear]{natbib}
\usepackage{url} 

\usepackage{amsthm,amsfonts,mathrsfs}
\usepackage{verbatim,amssymb,epsfig,enumitem,amsbsy,enumitem,color,bbm,mathtools}
\usepackage[hidelinks,colorlinks,citecolor=blue,urlcolor=blue]{hyperref}
\usepackage{subfigure}
\usepackage{multirow}
\usepackage{dsfont}
\usepackage{appendix}
\usepackage{algorithm}
\usepackage{algpseudocode}
\usepackage[modulo]{lineno}
\usepackage{xcolor}

\graphicspath{{figs/}}

\global\def\real{\mathbb R}
\global\def\expect{\mathbb E}
\global\def\var{\mathrm{var}}
\global\def\t{{\scriptscriptstyle\top}}
\global\def\prob{\mathbb P}
\global\def\stiefel{\mathcal S_{d,k}}

\global\def\data{\mathcal D}

\global\def\distconverge{\stackrel{d}{\to}}
\global\def\bp{\mathbf{p}}
\global\def\integer{\mathbb N}
\global\def\lop{\mathcal L}

\global\def\nx{n_x}
\def\nxtest{\bar{n}_{x}}
\global\def\nxfit{\tilde{n}_{x}}
\global\def\ny{n_y}
\def\nm{{n}}
\def\nplusm{n}
\global\def\nyfit{\tilde{n}_{y}}
\global\def\nytest{\bar{n}_{y}}
\def\mparam{m} 

\def\evmin{\Lambda_{\min}} 
\def\evmax{\Lambda_{\max}} 
\def\dataXfit{\mathcal X_{fit}}
\def\dataYfit{\mathcal Y_{fit}}
\def\IXfit{\tilde{\mathcal I}_X}
\def\IYfit{\tilde{\mathcal I}_Y}
\def\IXtest{\mathcal I_X}
\def\IYtest{\mathcal I_Y}
\def\mufit{\mu_n}
\def\nufit{\nu_n}
\def\kmax{k_\circ}
\def\pwrate{R}
\global\def\sparsestiefel{\mathcal{S}_{d,k, \tau}^{(1)}}
\global\def\sparsestiefelstar{\mathcal{S}_{d,k_\star, \tau_\star}^{(1)}}

\global\def\zerosparsestiefel{\mathcal{S}_{d,k, \varpi}^{(0)}}
\global\def\zerosparsestiefelstar{\mathcal{S}_{d,k_\star, \varpi_\star}^{(0)}}

\def\xytrunc{\theta_n}

\newtheorem{theorem}{Theorem}
\newtheorem{lemma}{Lemma}
\newtheorem{corollary}{Corollary}

\newtheorem{remark}{Remark}

\newtheorem{assumption}{Assumption}

\newcommand{\blind}{1}

\allowdisplaybreaks

\begin{document}
	
\date{}

\def\spacingset#1{\renewcommand{\baselinestretch}%
	{#1}\small\normalsize} \spacingset{1}


\if1\blind
{
	\title{\bf Neural Wasserstein Two-Sample Tests}
	\author{Xiaoyu Hu \\
	Department of Statistics and Data Science, Xi'an Jiaotong University \\
	and \\
	Zhenhua Lin \\
	Department of Statistics and Data Science, National Unversity of Singapore}
	\maketitle
} \fi

\if0\blind
{
	\bigskip
	\bigskip
	\bigskip
	\begin{center}
		{\LARGE\bf Neural Wasserstein Two-Sample Tests}
	\end{center}
	\medskip
} \fi

\bigskip

\begin{abstract}
	The two-sample homogeneity testing problem is fundamental in statistics and becomes particularly challenging in high dimensions, where classical tests can suffer substantial power loss.
	We develop a learning-assisted procedure based on the projection $1$-Wasserstein distance, which we call the \emph{neural Wasserstein test}.
	The method is motivated by the observation that there often exists a low-dimensional projection under which the two high-dimensional distributions differ.
	In practice, we learn the projection directions via manifold optimization and a witness function using deep neural networks. 
	To adapt to unknown projection dimensions and sparsity levels, we aggregate a collection of candidate statistics through a max-type construction, avoiding explicit tuning while potentially improving power.
	We establish the validity and consistency of the proposed test and prove a Berry--Esseen type bound for the Gaussian approximation.
	In particular, under the null hypothesis, the aggregated statistic converges to the absolute maximum of a standard Gaussian vector, yielding an asymptotically pivotal (distribution-free) calibration that bypasses resampling.
	Simulation studies and a real-data example demonstrate the strong finite-sample performance of the proposed method.
\end{abstract}
	
\noindent%
\vfill

\newpage
\spacingset{1.9} 

\section{Introduction}
\label{sec:intro}
The problem of the two-sample homogeneity hypothesis test is of fundamental importance in statistics and practice \citep{wald1940test,smirnov1948table,anderson1962distribution,bickel1969distribution,friedman1979,kim2020robust}, in which one aims to  determine whether two independent samples are sampled from the same distribution. This test finds wide application in various fields, such as genomics \citep{chen2010two} and finance \citep{horvath2013estimation}. 
Moreover, this test plays a crucial role in transfer learning, where it is often used to identify the nature of distributional shifts and to guide the selection of appropriate learning strategies \citep{lipton2018detecting,rabanser2019failing}. 
With the increasing prevalence of high-dimensional data, there is an increasing demand for methods tailored to such data. Our primary focus is on devising a valid and powerful test suitable for scenarios with diverging dimensions.

For the two-sample homogeneity test, one popular approach is based on kernel-based statistics \citep{sejdinovic2013equivalence}, including maximum mean discrepancy (MMD) \citep{gretton2006kernel} and energy distance \citep{szekely2004}. 
Due to their simplicity and wide applicability, these kernel-based methods, especially MMD, have attracted much attention in both methodological and theoretical developments \citep{gretton2012kernel,gretton2012optimal,sutherland2017generative,kirchler2020two}. 
To improve statistical power, \citet{liu2020learning} proposed a method involving the learning of a kernel function parameterized by deep neural nets on the training data, followed by the implementation of an MMD test using the learned kernel function and permutation on the test data. 
Additionally, to avoid permutations, \citet{kubler2022witness,shekhar2022permutation} utilized sample splitting to learn a one-dimensional function related to the kernel and construct a test statistic with a standard Gaussian limiting null distribution. 
These approaches are not particularly developed for high-dimensional contexts, which may result in limited power in high dimensions; see Section \ref{sec:sim}.
In high-dimensional settings, \citet{zhu2021interpoint} investigated the performance of kernel-based tests with permutation. \citet{gao2023two,yan2023kernel} introduced studentized MMD and studied the asymptotic properties as the dimension and the sample size diverge. {Their theoretical results highlight intrinsic limitations of kernel-based methods in power under high-dimensional regimes.}
Moreover, \citet{ramdas2015decreasing} demonstrated that the power of MMD can decay rapidly as the dimension increases under fair alternatives, a phenomenon that we also observe for other comparative methods in Section \ref{sec:sim}.

In contrast to kernel-based statistics, the Wasserstein distance has emerged as a powerful tool for measuring the discrepancies between distributions, thanks to its ability to preserve the underlying geometry of the distribution space, which has achieved success in many tasks, such as generative modeling \citep{arjovsky2017wasserstein}. 
%
However, the curse of dimensionality hinders its application in statistical estimation and hypothesis testing problems. For example, \citet{imaizumi2022hypothesis} attempted to employ the Wasserstein distance in the two-sample test, but the convergence rate therein deteriorates rapidly with increasing dimensionality. To mitigate the curse of dimensionality, several variants of the Wasserstein distance are proposed, including the smoothed Wasserstein distance \citep{goldfeld2020gaussian}, sliced or max-sliced Wasserstein distance and projection Wasserstein distance \citep{rabin2011wasserstein,deshpande2019max,lin2020projection}. In the case of low and fixed dimensions, there are works studying distributional limits of these variants and their applications in the two-sample test \citep{nietert2021smooth,goldfeld2022,sadhu2022,xi2022distributional}. 
However, these results do not extend to high-dimensional regimes.
The exploration of Wasserstein-based methods for high-dimensional two-sample testing remains relatively limited, apart from a small number of recent studies.
\citet{wang2021two} adopted the projected Wasserstein distance as a test statistic and employed a permutation test in practice. In parallel to our development, \citet{Hu2025} proposed a test based on the max-sliced Wasserstein distance combined with bootstrapping. 
However, these approaches can suffer from relatively low power due to the use of a single projection dimension and substantial computational overhead arising from resampling.

\paragraph*{Our contributions}
We develop a new learning-assisted and Wasserstein-based two-sample testing framework for high-dimensional data, with the following main contributions.
\begin{itemize}
	\item \textit{Projection-Wasserstein test statistic with deep neural network approximation.} Building on the projection 1-Wasserstein distance, we propose a two-sample test that learns (i) discriminative projection directions on the Stiefel manifold and (ii) an approximate 1-Lipschitz witness function using deep neural networks.
	\item \textit{Tuning-free adaptivity via max-type aggregation.} To adapt to unknown projection dimensions and sparsity levels, we aggregate statistics over a candidate set of hyperparameters using a max-type statistic, avoiding delicate tuning while potentially improving power. Theoretical power analysis under local alternatives are provided for the proposed test with and without sparsity regularizations.
	\item \textit{Pivotal limiting null and nonasymptotic guarantees.} We establish a Berry--Esseen type bound for Gaussian approximation and show that the aggregated test has a pivotal limiting null distribution, regardless of the ambient dimension and underlying distributions. This pivotal structure yields substantial computational savings compared to resampling-based methods such as permutation and bootstrap.
\end{itemize}

Specifically, the proposed test leverages the Kantorovich--Rubinstein duality to express the projection 1-Wasserstein distance 
between the distributions of $X, Y \in \real^d$ as
\begin{equation}\label{eq:PWD-dual}
	\sup_{f \in \mathcal F, U \in \stiefel} \{\expect f(U^\t X) - \expect f(U^\t Y)\},
\end{equation}
where $k$ is the projection dimension, $\mathcal F = \{ f: \real^k \to \real \mid f \text{ is 1-Lipschitz, } f(0)=0  \}$ and $\stiefel = \{U \in \real^{d\times k} : U^{\t}U = I_k \}$. In contrast with existing works \citep[e.g.,][]{wang2021two} utilizing a similar distance but bypassing estimation of the optimal $f$, 
our strategy is to estimate a discriminative function $f$ and  projection directions $U$ that achieve the supremum of \eqref{eq:PWD-dual}, and then construct a test statistic based on the estimates; 
our theoretical and numeric studies demonstrate the strength of this strategy. {Further distinctions between our approach and the method proposed by \citet{wang2021two} are discussed in Remark \ref{rmk:diff_wang}.}

Practically, to decouple the dependence caused by estimation, we employ data splitting to divide the available dataset into two disjoint subsets. The discriminative function and directions are trained on one subset, and the test statistic is constructed using the other subset.  The training phase involves estimating the projection directions $U$ and the 1-Lipschitz function $f$. However, optimizing both $f$ and $U$ jointly turns out to be too complex to achieve satisfactory numerical performance. To address this issue, we propose a practically effective two-step algorithm in which the optimization procedures for $U$ and $f$ are separated.
{For $U$, we consider both unregularized estimation and regularized estimation under $\ell_1$ and $\ell_0$ sparsity constraints to achieve adaptivity to unknown sparsity levels.} The non-smoothness and the manifold constraint motivate us to adopt the proximal gradient method for manifold optimization to obtain the discriminative $U$ \citep{chen2020proximal}.
For $f$, motivated by the excellent approximation power of deep neural networks, we adopt deep neural networks (where the number of layers grows with the sample size) to learn the discriminative function. 


The proposed test features several advantages. First, unlike data-driven calibration methods such as bootstrapping and permutation \citep{zhu2021interpoint,wang2021two},  the proposed pivotal test statistic enables a scalable computation. Second, in contrast with the studentized MMD,  {of which the  limiting null distributions depend on both  diverging dimension and sample size} \citep[see][for detailed discussions]{gao2023two,yan2023kernel}, our result holds irrespective of the dimensionality, underlying population or training algorithms; see Figures S.1 and S.2 for illustration. In practice, it could be challenging to determine whether the dimensionality is sufficiently large for applying the asymptotic null distribution of the studentized MMD test. 
Third, the proposed approach {adapts to unknown projection dimensions and sparsity levels through appropriate aggregation, thereby eliminating the intricate step of selecting the projection dimension and the sparsity parameter. It} has the potential to effectively capture the disparities between distributions via integrating information across multiple projection directions and regularization parameters in the proposed max-type statistic. The numerical results in Figure S.6 showcase the benefits of utilizing additional projections in some challenging scenarios. 

The rest of the paper is organized as follows. {In Section \ref{sec:statistic}, we introduce the proposed procedure for constructing the test statistic and analyse its size and power performance without regularization. Computation details and theoretical performance under $\ell_1$ and $\ell_0$ regularization are studied in Sections \ref{sec:l1} and \ref{sec:l0}, respectively.} Simulation studies are presented in Section \ref{sec:sim}, followed by the analysis of a real dataset in Section \ref{sec:real-data}. A brief concluding remark is provided in Section \ref{sec:conclusion}. More algorithmic details, simulation results and technical details are gathered in the Appendix and the Supplementary Material. 

\textbf{Notation.}
For a random variable $\xi$, the $\psi_2$-Orlicz norm is $\|\xi\|_{\psi_2} = \inf\{ t> 0, \expect(|\xi|^2/t^2) \le 1 \}$. 
For a deterministic vector $v \in \real^d$ and $0<p<\infty$, the $l_p$ norm is $\|v\|_p = (\sum_{j=1}^d |v_j|^p)^{1/p}$. Moreover, we write $\|v\|_\infty = \max_{j=1,\dots,d}|v_j|$ and $\|v\|_0 = |\{j: v_j \neq 0\}|$, where $|E|$ is the cardinality of a set $E$.  
For a matrix $A = (A_{ij}) \in \real^{d \times p}$, we write $\|A\|_F = (\sum_{i,j} A_{ij}^2)^{1/2}$, $\|A\|_2 = (\max_{\|v\|_2=1} v^\t AA^\t v)^{1/2} $, $\|A\|_\infty = \max_{i,j} |A_{ij}|$, $\|A\|_1 = \sum_{i,j} |A_{ij}|$ and $\|A\|_0 = |\{(i,j): A_{ij} \neq 0 \}|$. 
Let $\evmin(A)$ and $\evmax(A)$ represent the smallest and largest eigenvalues of a symmetric matrix $A$ respectively.
Let $B_r(x)$ denote the open ball of radius $r$ centered at $x$.
For a function $h: \real^d \to \real$, let $\|h\|_\infty = \sup_{x\in \real^d}|h(x)|$. {If for any $x, y \in \real^d$, $|h(x)-h(y)| \le \|x-y\|_2$, then $h$ is a 1-Lipschitz function.}
For two sequences $\{a_n\}$ and $\{b_n\}$ of non-negative real numbers, write $a_n \lesssim b_n$ if there is a constant $c>0$ and an integer $n_0 \ge 1$, not depending on $n$, such that $a_n \le cb_n$ for all $n \ge n_0$. Write $a_n \asymp b_n$ if $a_n \lesssim b_n$ and $b_n \lesssim a_n$.
The notation $a_n \ll b_n$ means that $a_n/b_n \to 0$ as $n \to \infty$.
For a set $\data$, we use $|\data|$ to denote its cardinality.

\section{A distribution-free test statistic}
\label{sec:statistic}

\subsection{Projection Wasserstein distance and two-sample tests} 
\label{subsec:pw}

Let $\mathcal P(\real^d)$ be the set of Borel probability measures on $\real^d$ with the finite first moment.
According to Definition 6.1 of \cite{villani2008optimal}, the 1-Wasserstein distance (hereafter, Wasserstein distance) between two probability measures $\mu$ and $\nu$ in $\mathcal P(\real^d)$ is 
\begin{equation*}\label{eq:wd-ot}
	W(\mu, \nu) = \inf_{\pi \in \Gamma(\mu, \nu)} \int \|x-y\|_2 d\pi(x,y) , 
\end{equation*}
where $\Gamma(\mu,\nu)$ represents the set of joint distributions on $\real^d \times \real^d$ with marginals $\mu$ and $\nu$. 

The Wasserstein distance, also known as the optimal transport distance, characterizes the minimal cost required to transport one distribution to another. It preserves the geometry of the space of distributions and finds wide applications in statistics and machine learning. 
Despite advantages of the Wasserstein distance, its empirical estimation suffers from the curse of dimensionality \citep{dudley1969speed,fournier2015rate,lei2020convergence}. 
One popular strategy to alleviate this issue is through projection \citep[e.g.,][]{lin2020projection,niles2022estimation}.

For $\mathscr X, \mathscr Y \subset \real^d$, $T: \mathscr X \to \mathscr Y$ and $\mu \in \mathcal P(\mathscr X)$, let $T_{\#}\mu$ denote the push-forward of $\mu$ by $T$, i.e., $T_{\#}\mu(A) = \mu(T^{-1}(A)) $ for any Borel set $A$ in $\mathscr Y$. Also, each element $U \in \stiefel$ is viewed as  a linear function $U(x) = U^\t x$ for $x \in \real^d$.  The $k$-dimensional projection  Wasserstein distance \citep{paty2019subspace,niles2022estimation} between $\mu$ and $\nu$ is 
\begin{equation}\label{eq:primal}
	PW_k(\mu, \nu) = \sup_{U\in \stiefel} W(U_{\#}\mu, U_{\#}\nu) = \sup_{U \in \stiefel} \inf_{\pi \in \Gamma(\mu, \nu)} \int \|U^\t(x-y)\|_2 d\pi(x,y). 
\end{equation}
The following Lemma \ref{lem:properties}, due to  \citet{paty2019subspace}, asserts that $PW_k$ is a proper distance over $\mathcal P(\real^d)$. 
\begin{lemma}[Paty and Cuturi, 2019]\label{lem:properties} The $PW_k$ is a distance over $\mathcal P(\real^d)$. Also, for $\mu, \nu \in \mathcal P(\real^d)$, there exists a $U_0 \in \stiefel$ such that 
	\[ PW_k(\mu, \nu) = W(U_{0\#}\mu, U_{0\#}\nu). \]
\end{lemma}

From Lemma \ref{lem:properties}, we conclude that $PW_k(\mu, \nu)$ vanishes if and only if $\mu=\nu$. Moreover, the supremum in \eqref{eq:primal} can be achieved, and we call $U_0$ an optimal projection.

{The projection Wasserstein distance $PW_k(\mu, \nu)$ is the same as the Wasserstein distance $W(\mu, \nu)$ if the two distributions $\mu$ and $\nu$ satisfy the spiked transport model with the subspace dimension $k$, see Proposition 3 in \citet{niles2022estimation}.}

Given two independent random samples in $\real^d$, $X_1, \dots X_{\nx} \stackrel{i.i.d.}{\sim} \mu$ and $Y_1, \dots, Y_{\ny} \stackrel{i.i.d.}{\sim} \nu$, the goal is to test the hypothesis
\[ H_0 : \mu=\nu\qquad\text{versus}\qquad H_1: \mu \neq \nu, \]
where $\mu, \nu \in \mathcal P(\real^d)$ are two underlying distributions. In light of Lemma \ref{lem:properties}, the null hypothesis is equivalent to $PW_k(\mu, \nu)=0$, based on which one can devise a test by using the sample version of $PW_k(\mu, \nu)$ \citep[e.g.,][]{wang2021two}. One difficulty of such test is to accurately estimate the p-value or critical value.

\subsection{Overview of the proposed new two-sample test}
\label{subsec:method}
Instead of directly adopting the distance of \eqref{eq:primal} and its sample version, we first utilize Lemma \ref{lem:properties} and the Kantorovich-Rubinstein theorem \citep{villani2003topics} to observe
\begin{equation}\label{eq:U0-dual}
	PW_k(\mu, \nu) =  W(U_{0\#}\mu, U_{0\#}\nu)= \sup_{f\in \mathcal F} \{\expect f(U_0^{\t}X) - \expect f(U_0^{\t}Y)\}, 
\end{equation}
where $X \sim \mu, Y \sim \nu$, $\mathcal F = \{ f: \real^k \to \real \mid f \text{ is 1-Lipschitz, } f(0)=0  \}$ and 
$U_0$ is the optimal projection in Lemma \ref{lem:properties}.
If $PW_k(\mu, \nu)\neq 0$,  
then $\sup_{f\in \mathcal F} \{\expect f(U_0^{\t}X) - \expect f(U_0^{\t}Y)\}\neq 0$ and consequently there exists a function $f_0\in \mathcal F$ such that $\expect g_0(X) - \expect g_0(Y) \neq  0$ for $g_0=f_0\circ U_0$. Conceptually, we would choose $f_0$ such that $\expect g_0(X) - \expect g_0(Y)$ is as close to $PW_k(\mu, \nu)$ as possible, as discussed in Appendix \ref{apx:f0}. It is intuitive to construct a test statistic based on the sample version of $\sup_{f\in \mathcal F} \{\expect f(U_0^{\t}X) - \expect f(U_0^{\t}Y)\}$ and estimates of $U_0$.  Moreover, the representation in \eqref{eq:U0-dual} motivates an effective two-step procedure, where $U_0$ is estimated first, followed by the estimation of $f_0$. The full estimation procedure is detailed in Section~\ref{subsec:alg}, and the resulting test statistic is presented in Section~\ref{subsec:max-statistic}. It turns out that, estimating the null distribution of this test statistic is much more tractable, making it easier to compute the p-value or critical value.

There are two difficulties in finding an estimate $\hat f$ of the oracle discriminative function $f_0$ and an estimate $\hat U$ of  the direction $U_0$ from observed data. 
First, the Lipschitz class is too large to ensure the estimation accuracy.
Second, it is hard to construct a pivotal test statistic due to the complex correlation among the quantities $\hat f(\hat U^\t X_i)$ and $\hat f(\hat U^\t Y_j)$ for $i=1,\dots, \nx$ and $j=1, \dots, \ny$.  
To tackle these issues, we resort to a suitably chosen function class for estimation and utilize the sample splitting technique to decouple the dependence. Below we elaborate these ideas.

For a suitable function class, we consider estimating $f_0$ using neural networks with rectified linear unit (ReLU) activation functions, which have successful applications in various fields and have gained increasing attention in statistics \citep{schmidt2020nonparametric,zhong2022cox,fan2023,yan2023npdnn}. Specifically, consider the network architecture $(D, \bp)$ consisting of $D \in \integer$ hidden layers and a width vector $\bp=(p_0, \dots, p_{D+1}) \in \integer^{D+2}$. Let $A_{\ell} \in \real^{p_{\ell+1} \times p_\ell}$ and $b_{\ell} \in \real^{p_{\ell+1}}$ be a weight matrix and a bias vector, respectively, for the $\ell$-th layer. Let $\sigma(x) = \max(x, 0)$ be the ReLU activation function. A neural network with the network architecture $(D, \bp)$ is any function of the following form
\begin{equation}\label{eq:func-nn}
\phi : \real^{p_0} \to \real^{p_{D+1}}, \quad \phi(x)  =  \lop_D \circ \sigma \circ \lop_{D-1} \circ \sigma \cdots \sigma \circ \lop_1 \circ \sigma \circ \lop_0(x),
\end{equation} 
where $\lop_\ell(x) = A_\ell x + b_\ell$ for $\ell =0, \dots D.$ The width $L$ is defined as the maximum width of the hidden layers, i.e., $L = \max\{p_1, \cdots, p_D\}$. 
In this paper, we take the following functional class
\begin{align*}
\mathcal N(D, \bp, s, B) := \bigg\{ \phi \text{ of the form \eqref{eq:func-nn}}: & \max_{\ell=0, \dots, D} \|A_\ell\|_\infty \vee \|b_\ell\|_\infty \le 1, \\
& \sum_{\ell=0}^D \|A_\ell\|_0 + \|b_\ell\|_0 \le s, \|\phi\|_\infty \le B    \bigg\}.
\end{align*}

To overcome the complex correlation among the estimators, we employ sample splitting to decouple the dependence. Specifically, we randomly partition the indices $1,\ldots,\nx$ into disjoint $\IXfit$ and $\IXtest$ such that $\{1,\ldots,\nx\}=\IXfit\cup\IXtest$, and partition the indices $1,\ldots,\ny$ into $\IYfit$ and $\IYtest$ such that $\{1,\ldots,\ny\}=\IYfit\cup\IYtest$. With these partitions, we estimate $f_0$ and $U_0$ by using the observations $\data_{fit} = \dataXfit \cup \dataYfit$, where $\dataXfit = \{ X_i, i \in \IXfit\}$ and $\dataYfit = \{ Y_j, j \in \IYfit\}$, and then construct a test statistic based on the observations $\data_{test} = \{X_i, i \in \IXtest\} \cup \{Y_j, j \in \IYtest\}$. This sample-splitting strategy ensures that the estimators $\hat{f}$ and $\hat{U}$, obtained from the training subset, are independent of the testing subset used to construct the test statistic. Consequently, conditional on $\hat{f}$ and $\hat{U}$, the key quantity $\hat{S}_n$ (to be defined in \eqref{eq:Sng}) becomes a weighted average of independent random variables. This decoupling of the complex dependence between the estimators and the test statistic is crucial for establishing the validity of our test and enables a rigorous power analysis.

We are now ready to present the main idea of the proposed approach. By using the fitting data $\data_{fit}$, we estimate $U_0$ and employ deep neural networks in $\mathcal N(D, \bp, s, B)$ to estimate  $f_0$; details are provided in Section \ref{subsec:alg}.	
Once we obtain the estimator $\hat g = \hat f \circ \hat U$, we  construct the test statistic on the rest of the data $\data_{test}$, which is detailed in Section \ref{subsec:max-statistic}.


\subsection{An effective algorithm for computing $(\hat f,\hat U)$}
\label{subsec:alg}
The representation in \eqref{eq:U0-dual} motivates us to estimate $U_0$ and $f_0$ separately in two steps. In the first step, by utilizing the primal formulation in \eqref{eq:primal}, we estimate $U_0$ without involving the function space of $f_0$.
After obtaining $\hat U$, in the second step we train the discriminative function by fitting a deep neural network on the projected data $\{\hat U^\t X: X\in\dataXfit\}$ and $\{\hat U^\t Y: Y\in \dataYfit\}$. Below we outline these two steps.

\textbf{Step 1.} Let $\nxfit=|\IXfit|$ and $\nyfit=|\IYfit|$, and define
\begin{equation}\label{eq:hat-mu-nu}
\mufit  = \frac{1}{\nxfit}\sum_{i\in\IXfit} \delta_{X_i} \qquad \text{and}\qquad \nufit = \frac{1}{\nyfit}\sum_{j\in \IYfit} \delta_{Y_j},
\end{equation}
where $\delta_x$ denotes the Dirac delta function at $x$.
{We first describe a generic framework for estimating $U_0 \in \mathcal U(\varrho)$ for subset $U(\varrho)$, where $U(\varrho)$ is a subset of $\stiefel$ and determined by a regularization parameter $\varrho$ (e.g., sparsity level). This leads to the following estimator} 
\begin{align}\label{EQ:HATU}
\hat U\in \mathop{\arg\max}\limits_{U \in \mathcal U(\varrho)}\bigg\{ \min\limits_{\pi \in \Gamma(\mufit, \nufit)} \sum_{i\in\IXfit}\sum_{j\in \IYfit} \pi_{ij}\|U^{\t}(X_{i}-Y_{j})\|_2\bigg\},
\end{align}
where 
and
\[ \Gamma(\mufit, \nufit) = \left\{ \pi \in \real^{\nxfit \times \nyfit}: \sum_{i \in \IXfit} \pi_{ij} = \frac{1}{\nyfit}\text{ and } \sum_{j\in \IYfit} \pi_{ij} = \frac{1}{\nxfit} \right\}. \]
{While a natural choice of $\mathcal U(\varrho)$ is $\stiefel$, additional structural constraints can be incorporated to help achieve adaptivity and favorable performance in high dimensions. Specific choices of the constraint set $\mathcal U(\varrho)$ depending on some regularization parameter $\varrho$ and the corresponding estimation algorithms are developed later.}

\textbf{Step 2.} After obtaining an estimator $\hat U$ in the first step, we use deep neural networks to train the discriminative function on $\data_{fit}$. Specifically, we find the estimate
\begin{equation}\label{eq:hatf}
\hat f \in \mathop{\arg\max}_{f \in \mathcal N(D, \bp, s,B)} \left\{ \frac{1}{\nxfit}\sum_{i\in\IXfit} f(\hat U^{\t}X_{i}) - \frac{1}{\nyfit}\sum_{j\in \IYfit} f(\hat U^{\t}Y_{j})\right\}.
\end{equation} 
For the network parameters $D$, $\mathbf p$, $s$ and $B$, the theoretical analysis in Theorem \ref{thm:max_size} suggests the choices of these parameters do not have much influence on the size of the proposed test. {In constrast, the power analysis demonstrates that these neural network parameters do affect the test power, and practical guidelines for their selection are provided later.}

\begin{remark}\label{rmk:f0}
Throughout this paper, $f_0$ represents the conceptual oracle discriminative function, while $f$ denotes a generic function from a specified function class. We emphasize that the oracle discriminative function $f_0$ belongs to the class of 1-Lipschitz functions $\mathcal F = \{ f: \real^k \to \real \mid f \text{ is 1-Lipschitz, } f(0)=0  \}$. We do not assume that $f_0$ itself lies within the class of neural networks. Neural networks are employed solely to approximate this class and obtain the estimator $\hat f$.
\end{remark}

\begin{remark}
Instead of the proposed two-step procedure, one might consider estimating $\hat{f}$ and $\hat{U}$ simultaneously by solving
\begin{equation}\label{eq:joint_opt}
	(\hat f,\hat U)\in\underset{f\in \mathcal N(D, \bp, s,B),U\in \mathcal U(\varrho)}{\arg \max}\left\{ \frac{1}{|\IXfit|}\sum_{i\in \IXfit} f(U^{\t}X_i) - \frac{1}{|\IYfit|}\sum_{j\in\IYfit} f(U^{\t}Y_j)\right\}.
\end{equation}
However, in practice, optimizing \eqref{eq:joint_opt} poses significant challenges. For instance, we find that directly optimizing both $f$ and $U$ using an alternating algorithm often fails to achieve satisfactory power.
\end{remark}

\subsection{The test statistic via aggregating multiple projection dimensions and regularization parameters}\label{subsec:max-statistic}

After obtaining $\hat U$ and $\hat f$, based on the rest of the data $\data_{test}$, one may construct the following test statistic 
\begin{equation*}\label{eq:stat}
\left(\frac{\nxtest\nytest}{\nxtest+\nytest}\right)^{1/2}
\frac{S_n(\hat g)  }{\hat \sigma},
\end{equation*}
where $\nxtest=|\IXtest|$ and $\nytest=|\IYtest|$ are respectively the number of observations of $X$ and $Y$ in the test data $\data_{test}$, and
\begin{align}
S_\nm(\hat g)& = \frac{1}{\nxtest}\sum_{i\in\IXtest} \hat g(X_i) - \frac{1}{\nytest}\sum_{j\in\IYtest} \hat g(Y_j),\label{eq:Sng}\\
\hat \sigma^2 & = \frac{\nytest}{\nxtest+\nytest}\frac{1}{\nxtest} \sum_{i\in\IXtest} \big( \hat g(X_i) -\expect_{\nxtest} (\hat g(X) ) \big)^2 + \nonumber \\
&\qquad \frac{\nxtest}{\nxtest+\nytest} \frac{1}{\nytest} \sum_{j\in \IYtest} \big( \hat g(Y_j) -\expect_{\nytest} (\hat g(Y) ) \big)^2;\nonumber
\end{align}
in the above, $\expect_{\nxtest} \hat g(X)   = \nxtest^{-1}\sum_{i\in\IXtest} \hat g(X_i)$ and $\expect_{\nytest} \hat g(Y)   = \nytest^{-1}\sum_{j\in\IYtest} \hat g(Y_j)$.

A crucial issue is the choice of the projection dimension $k$ and the regularization parameter $\varrho$ for estimating the optimal projection. While these parameters do not affect the size, as evidenced in Theorem~\ref{thm:max_size}, they influence the power performance. 
To get around this delicate issue, we propose to aggregate strength from multiple values of $(k, \varrho)$ via a max-type test statistic, as follows.  
Given a candidate set of hyperparameters $\mathcal C = \{ (k_j, \varrho_j), j=1, \dots, \mparam \}$, let $\hat S_n = (S_n(\hat g_1), \dots, S_n(\hat g_\mparam))^\t$, where $\hat g_j$ is the corresponding estimator using the hyperparameter pair $(k_j, \varrho_j)$ for $j=1,\dots, \mparam$. 
Instead of selecting a value of $j$,  we aggregate the test statistics corresponding to $\mathcal C$ into the following final max-type test statistic,
\begin{equation}\label{eq:max_stat}
T_\nm =  \max_{j=1,\dots,\mparam}  \bigg(\frac{\nxtest\nytest}{\nxtest+\nytest}\bigg)^{1/2} |e_j^{\t} \hat{\Sigma}^{-1/2} \hat S_\nm|,
\end{equation}
where $\nplusm=\nx+\ny$ is the total sample size, $e_j \in \real^\mparam$ has 1 in the $j$th coordinate and 0 otherwise, and $\hat \Sigma$ is the sample estimate of 
\begin{equation}\label{eq:Sigma}
\Sigma = \var\{(\nxtest\nytest/(\nxtest+\nytest))^{1/2}\hat S_\nm \mid \data_{fit}\}. 
\end{equation}
Given a significance level $\alpha\in (0,1)$, we reject the null hypothesis if $T_\nm > q_{1-\alpha}$, where $q_{1-\alpha}$ is the $1-\alpha$ quantile of $\max_{j=1,\dots,\mparam} |Z_j|$ with $Z \sim N(0, I_\mparam)$, as demonstrated in Theorem \ref{thm:max_size}.

\begin{remark}\label{rmk:two-sided}
Although Equation \eqref{eq:U0-dual} may  suggest a one-sided  test, we find it beneficial to employ a two-sided test when aggregating multiple test statistics obtained from several candidate parameters. One reason is that the standardization by $\hat\Sigma^{-1/2}$ rotates the vector $\hat S_n$, which may alter the signs of the entries of $\hat S_n$. Adopting a one-sided test in this case  leads to potential loss of power and also adds difficulties to our theoretical analysis.
\end{remark}

\begin{remark}\label{rem:Sigma}
To ensure $\hat\Sigma^{-1}$ is numerically stable, we require the smallest eigenvalue of  $\hat\Sigma$ to be no smaller than some constant $\kappa>0$.
This can be achieved by selecting a suitable set $\mathcal C$ of candidate hyperparameters. For example, in practice, if we find that $\hat\Sigma$ is near singular, then we can adopt the following backward elimination procedure to reduce the size of $\mathcal C$. Let $\hat\Sigma_{-j}$ be the sample covariance after excluding $(k_j,\rho_j)$. Then we find $j^\prime= \arg\max_{j}\evmin(\hat\Sigma_{-j})$ and then remove $(k_{j^\prime},\rho_{j^\prime})$ from $\mathcal C$. Continue this process recursively until $\hat\Sigma$ is well conditioned.  Our numeric experiments show that this procedure practically works well.
\end{remark}

\begin{remark}\label{rem:C-and-k}
In practice, a small set $\mathcal C$ of candidate hyperparameters and small projection dimensions $k$ are often sufficient, i.e., with a bounded $m$ and a bounded $\kmax=\max\{k: (k, \varrho)\in\mathcal C\}$. Note that 
a larger value of $k$ does not necessarily lead to higher power. This is because the estimation variability of $\hat f$ increases with $k$, which leads to greater variability in the test statistic and negatively impacts the power. 
\end{remark}

\subsection{Limiting null distribution and size analysis}

We begin with stating some assumptions required for theoretical investigation. 
Assumption \ref{assump:sample_size} merely requires the ratio of sample sizes $\nx/\ny$ to be bounded away from zero and infinity, which is weaker than the standard convergent assumption \citep{pan2018ball,zhu2021interpoint,yan2023kernel}.
The Assumption \ref{assump:pd} is a mild condition on the spectrum of the covariance matrix $\Sigma$ in \eqref{eq:Sigma}. In particular, its smallest eigenvalue can be practically lower bounded away from zero via selecting a suitable set of candidate hyperparameters; see Remark \ref{rem:Sigma}. Recall that $\nplusm$ represents the total sample size, and $\nxfit$ and $\nyfit$ are respectively the number of observations of $X$ and $Y$ for estimating $(\hat U,\hat f)$.

\begin{assumption}\label{assump:sample_size}
For universal constants $0<c_1\le c_2 <1$, $\nx/(\nx+\ny) \in [c_1, c_2]$, $\nxfit/\nx \in [c_1, c_2]$ and $\nyfit/\ny\in [c_1, c_2]$.
\end{assumption}

\begin{assumption}\label{assump:pd} There are constants $0<c'_1, c'_2<\infty$ such that,
conditional on $\data_{fit}$, the covariance matrix $\Sigma$ in \eqref{eq:Sigma} satisfies $c'_1 \le \evmin(\Sigma) \le \evmax(\Sigma) \le c'_2$. 
\end{assumption}

Below we use $\expect_\star$ and $\var_\star$ to denote the conditional expectation and conditional variance given the fitting sample, respectively. For example, $\expect_\star\hat g(X)= \expect(\hat g(X) \mid \data_{fit} )$. 
In Theorem~\ref{thm:max_size}, we establish the validity of the proposed test statistic.

\begin{theorem}\label{thm:max_size}
Suppose that Assumptions \ref{assump:sample_size} and \ref{assump:pd} hold.	Under the null hypothesis, $$T_\nm \distconverge \max_{j=1,\dots,\mparam} |Z_j|, ~~ \text{ as } \nplusm \to \infty,$$ where $Z = (Z_1, \dots, Z_\mparam)^\t  \sim N(0, I_\mparam)$ and $\distconverge$ stands for convergence in distribution. 

Moreover, if we assume there exists an $M_{\nm}$ such that $\expect_\star \exp(\hat g_j^2(X) / M_{\nm}^2) \le 2$ and $\expect_\star \exp(\hat g_j^2(Y) / M_{\nm}^2) \le 2$ almost surely for $j=1,\dots,\mparam$, then we have
\[  \sup_{t \in \real} \bigg|\prob(T_{\nm} \le t) - \prob\bigg(\max_{j=1,\dots, \mparam} |Z_j| \le t\bigg) \bigg| \le c\frac{ G_n }{n^{1/2} } + c \frac{M_{\nm}^2 \log n}{n^{1/2}}  \]
for some positive constant $c>0$, where $G_n=\max_{j=1,\dots,m}\expect (|\hat g_j(X)|^3 + |\hat g_j(Y)|^3).$
\end{theorem}

\begin{remark}\label{rmk:berry-essen}
The sub-Gaussian assumption on $\hat g_j(X)$ is not stringent. Due to the structure of the neural network $\mathcal N(D, \bp, s,B)$, we have $\hat g_j(X) \le B$ and thus $M_n \le 2B$. 
Since bounding the sub-Gaussian norm by the upper bound $B$ could be overly conservative, we may expect $M_n$  to be often much smaller than $B$. In the case of bounded $B$, the rate of convergence in the Berry-Essen type bound of Theorem \ref{thm:max_size} is $1/\sqrt{n}$ (up to a $\log n$ factor).
Otherwise, we may adopt the Lipschitz-constrained neural networks in \citet{anil2019sorting}. For such neural networks, the Lipschitz property implies that $G_n\leq c$ and  $M_n \le c$ for some constant $c>0$ under Assumption \ref{assump:tail} when $k_j$ is bounded, according to Lemma S.22. In this case, we can still achieve the same near-parametric rate of convergence in Theorem \ref{thm:max_size}. In addition, the power analysis can be extended to accommodate the Lipschitz-constrained neural network via modifying the proofs; we do not pursue this direction in order to maintain a focused presentation and to numerically utilize existing computational packages designed for ordinary ReLU networks.
\end{remark}

From Theorem \ref{thm:max_size}, the proposed test statistic $T_n$ is asymptotically distribution-free in the sense that the limiting null distribution corresponds to the absolute maximum of a standard Gaussian random vector, independent of the underlying population. The pivotal test statistic avoids the use of resampling/permutation methods \citep{zhu2021interpoint}, and thus eases the computation. The asymptotic distribution is established regardless of the estimation algorithms (e.g., network structures) and the dimensionality $d$. In contrast, the distributional limits of MMD-based statistics depend on the asymptotic regimes {of dimensionality and the sample size} \citep{gretton2012kernel,gao2023two,yan2023kernel}. 

Moreover, we establish the Berry--Essen type bound for Gaussian approximation in Theorem \ref{thm:max_size}. The resulting convergence rate depends on the moment conditions and tail behavior of $\hat g_j(X)$ and $\hat g_j(Y)$, which are further discussed when we specify the constraint set and neural network structures in Sections~\ref{sec:l1} and~\ref{sec:l0}.

\begin{remark}
The estimators $\hat U$ and $\hat f$ may not be unique, but the results in Theorem \ref{thm:max_size} in the above and Theorem \ref{thm:max_power_l1} or \ref{thm:max_power_l0} below remain valid for any choices of $\hat U$ and $\hat f$ that respectively maximize  \eqref{EQ:HATU} and \eqref{eq:hatf}. This is because our test statistic is built on the maximum rather than the maximizers, so that any maximizer plays the same role.
\end{remark}

\subsection{Power analysis}\label{subsec:power}
Although the test size is insensitive to the estimation accuracy of $\hat U$ and $\hat f$, the power depends on their estimation quality. In this section, we study the power performance under the choice $\mathcal U(\varrho) = \stiefel$ in \eqref{EQ:HATU}, which corresponds to estimating of projection directions without imposing any regularization. Consequently, the test statistic $T_n$ is constructed by aggregating information across multiple projection dimensions, with $\mathcal C = \{k_1, \cdots, k_m\}$.

\begin{assumption}\label{assump:tail}
$\sup_{\|u\|_2=1} \|X^{\t}u\|_{\psi_2} \le \eta$ and $\sup_{\|u\|_2=1} \|Y^{\t}u\|_{\psi_2} \le \eta$ for some positive constant $0 < \eta < \infty$.
\end{assumption}

The sub-Gaussian tail condition in Assumption \ref{assump:tail} is common in high-dimensional statistics \citep{van2014asymptotically,lopes2022central}.

\begin{lemma}\label{lem:truncation_Fn}
Suppose Assumption 3 holds. Define the event 
\begin{equation*}\label{eq:trunc_event}
	\begin{aligned}
		E_n = \{ & \|X_i\|_\infty \le \eta\sqrt{2\log(dn)} \text{ for }  1 \le i \le \nx \text{ and } \\
		& \|Y_j\|_\infty \le \eta\sqrt{2\log(dn)} \text{ for } 1 \le j \le \ny \}.
	\end{aligned}
\end{equation*}
We have $\prob(E_n) \ge 1 - 2n^{-1}$.	
\end{lemma}

The class $\mathcal F$ is not compact under the norm $\|\cdot\|_\infty$, which prevents a direct application of empirical process techniques in the power analysis. To address this issue, Lemma \ref{lem:truncation_Fn} suggests that it suffices to restrict attention to the smaller class $\mathcal F_{n} = \{ f\in \mathcal F: f(x) = f(\tilde x), \text{ where } \tilde x = P_{[-d^{1/2} \xytrunc , d^{1/2} \xytrunc ]^k}(x) \}$, where $\xytrunc = \eta \sqrt{2\log (dn)}$ and $P_{[-a , a]^k}(x)$ denotes the {projection of $x$ onto $[-a, a]^k$}, i.e., the closest point to $x$ within $[-a, a]^k$. 

Now we are ready to present the results on power performance. Let 
$$\pwrate_n(k,d)=(kd)^{\frac{1}{2}}(\log(dn))^{\frac{3}{2}}\{n^{-\frac{1}{k+2}}+ n^{-\frac{1}{2}}(k d)^{\frac{1}{2}} + n^{-\frac{2}{k+2}}\log(n) + n^{-1}\log(n)kd\}.$$
\begin{theorem}\label{thm:max_power}
Suppose Assumptions \ref{assump:sample_size}, \ref{assump:pd} and \ref{assump:tail} hold. For any $k \in \mathcal C$, take the corresponding ReLU neural network to be
$\mathcal N(D, (k, 2k, 6(k+1)N, \dots, 6(k+1)N, 1), s, 4(kd)^{1/2}\xytrunc)$, 
where $D \asymp \log n$, $s \asymp N D$, $N \ge 2^k \vee(4(kd)^{1/2}\xytrunc+1)e^k$ and 
\[ N \asymp n^{\frac{k}{2+k}} (\log n)^{-\frac{3k}{2+k}}.\]
If there exists a $k_\star \in \mathcal C$ such that
\begin{equation}\label{eq:power_signal}
	PW_{k_\star}(\mu, \nu) \ge c_a R_n(k_\star,d)
\end{equation}
for some sufficiently large constant $c_a>0$, then we have $\prob(T_n > q_{1-\alpha}) \to 1$ as $n \to \infty$.	
\end{theorem}

\begin{remark}\label{rmk:composition}
If the optimal function $f_0$ exhibits some low-dimensional structure, such as the hierarchical structure in \citet{schmidt2020nonparametric}, the dependence on $k_\star$ in the convergence rate $n^{-1/(k_\star+2)}$ can be further improved. We leave the exploration of this direction for future work.
\end{remark}

\begin{remark}\label{rmk:nn}
{In addition to the sparse neural network structure employed in our paper, alternative structures such as fully connected feedforward neural networks with ReLU activation functions \citep{kohler2021rate,jiao2023deep} might be adopted. We believe that analogous theoretical guarantees could potentially be established for these fully connected networks, and we leave a detailed investigation of this direction to future work.}
\end{remark}

In the above theorem, \eqref{eq:power_signal} implies that $\mu\neq \nu$, i.e., the alternative hypothesis is true. The power analysis can be extended to accommodate growing $m$ and $\kmax$ defined in Remark \ref{rem:C-and-k}, with additional technicalities. However, we opt to not pursue this direction in light of Remark \ref{rem:C-and-k}. The network structure in the theorem is chosen to balance the estimation error and the approximation error \citep{schmidt2020nonparametric}.

Theorem \ref{thm:max_power} suggests potential increase of power by incorporating multiple projection dimensions $k$. For instance, one $k_\star$ with $PW_{k_\star}$ satisfying \eqref{eq:power_signal} is already sufficient for consistency, and multiple such $k_\star$ increase the chance of detecting the distributional discrepancy. In addition, since $PW_{k_1}(\mu, \nu) \ge PW_{k_2}(\mu, \nu)$ for $k_1 \ge k_2$ according to Lemma \ref{lem:monotonicity_k}, if $PW_{k_\star}$ also satisfies \eqref{eq:power_signal} for a larger $k_\star$, then the signal is amplified and potentially becomes relatively easier to detect, which leads to higher  power.

\begin{lemma}\label{lem:monotonicity_k}
The projection  Wasserstein distance is non-decreasing with respect to the projection dimension $k$, that is, $PW_{k_1}(\mu, \nu) \ge PW_{k_2}(\mu, \nu)$ if $k_1 \ge k_2$.
\end{lemma}

The result in Theorem \ref{thm:max_power} is nonasymptotic in both $n$ and $d$, implying that no additional restrictions on the relationship between $d$ and $n$ are required to ensure consistency beyond the signal condition in \eqref{eq:power_signal} . 
Nevertheless, in high dimensions, the distributional differences are often concentrated on a small subset of coordinates, which motivates introducing sparsity in the subspace projection to enhance power. Accordingly, we study the power performance under $\ell_1$ and $\ell_0$ regularization, respectively. 
To avoid notational ambiguity, we use distinct notations for the corresponding regularization parameters when specifying the constraint set $\mathcal U(\cdot)$.

\section{Power enhancement with $\ell_1$ regularization}\label{sec:l1}
Let the constraint set for the projection in \eqref{EQ:HATU} be $\sparsestiefel = \{ U \in \stiefel: \|U\|_1 \le \tau \}$ where $k \le \tau \le kd^{1/2}$.

\subsection{Estimation and the test statistic}\label{subsec:statistic_l1}
Under $\ell_1$ regularization, the optimization problem in \eqref{EQ:HATU} becomes
\begin{align}\label{EQ:HATU_l1}
\hat U_{k, \tau}^{(1)} \in \mathop{\arg\max}\limits_{U \in \sparsestiefel} \bigg\{ \min\limits_{\pi \in \Gamma(\mufit, \nufit)} \sum_{i\in\IXfit}\sum_{j\in \IYfit} \pi_{ij}\|U^{\t}(X_{i}-Y_{j})\|_2 \bigg\}.
\end{align}
A smaller $\tau$ corresponds to sparser projection directions, reducing the complexity of the parameter space and thereby simplifying estimation. 
Given $\hat U_{k, \tau}^{(1)}$, we then use the form in \eqref{eq:hatf} from Step 2 to obtain the estimator $\hat f_{k, \tau}^{(1)}$, with the network parameters suggested in Thereom~\ref{thm:max_power_l1}. 

Given a candidate set of hyperparameters $\mathcal C^{(1)} = \{ (k_j, \tau_j), j=1, \dots, \mparam \}$, let $\hat S_{n}^{(1)} = (S_n(\hat g_{1}^{(1)}), \dots, S_n(\hat g_{\mparam}^{(1)}))^\t$, where $\hat g_{j}^{(1)}$ is the corresponding estimator using the hyperparameter pair $(k_j,\tau_j)$ for $j=1,\dots, \mparam$.  
The corresponding max-type test statistic is given by
\begin{equation}\label{eq:max_stat_l1}
T_{\nm}^{(1)} =  \max_{j=1,\dots,\mparam}  \bigg(\frac{\nxtest\nytest}{\nxtest+\nytest}\bigg)^{1/2} |e_j^{\t} (\hat{\Sigma}^{(1)})^{-1/2} \hat S_{\nm}^{(1)}|,
\end{equation}
where $\hat \Sigma^{(1)}$ is the sample estimate of 
\begin{equation}\label{eq:Sigma_l1}
\Sigma^{(1)} = \var\{(\nxtest\nytest/(\nxtest+\nytest))^{1/2}\hat S_{\nm}^{(1)} \mid \data_{fit}\}. 
\end{equation}
Given a significance level $\alpha\in (0,1)$, we reject the null hypothesis if $T_{\nm}^{(1)} > q_{1-\alpha}$, where $q_{1-\alpha}$ is the $1-\alpha$ quantile of $\max_{j=1,\dots,\mparam} |Z_j|$ with $Z \sim N(0, I_\mparam)$.

\begin{corollary}\label{cor:size_l1}
Suppose that Assumption \ref{assump:sample_size} holds and $\Sigma^{(1)}$ satisfies Assumption \ref{assump:pd}. Under the null hypothesis, we have $T_\nm^{(1)} \distconverge \max_{j=1,\dots,\mparam} |Z_j|$, as $\nplusm \to \infty$, where $Z \sim N(0, I_\mparam)$. Moreover, with the network structures specified in Theorem \ref{thm:max_power_l1}, we have
\[  \sup_{t \in \real} \bigg|\prob(T_{\nm}^{(1)} \le t) - \prob\bigg(\max_{j=1,\dots, \mparam} |Z_j| \le t\bigg) \bigg| \le c\frac{ k_\circ^{\frac{3}{2}}(\eta\tau_\circ)^3 (\log(dn))^{\frac{3}{2}} + k_\circ(\eta\tau_\circ)^2\log(dn)\log(n)}{n^{1/2} }  \]
for some positive constant $c>0$, where $k_\circ = \max\{k: (k, \tau)\in \mathcal C^{(1)}\}$ and $\tau_\circ = \max\{\tau: (k, \tau)\in \mathcal C^{(1)}\}$.
\end{corollary}

The weak convergence of the proposed test statistic under the null hypothesis ensures asymptotic type I error control at a prespecified level $\alpha$. Moreover, Corollary \ref{cor:size_l1} establishes the Berry--Essen type bound for the Gaussian approximation of the test statistic under $\ell_1$ regularization, with convergence rate of order $n^{-1/2}\log(dn)\{ (\log(dn))^{1/2} + \log(n) \}$ if $\tau_\circ \lesssim 1$. 

\subsection{Power analysis}\label{subsec:power_l1}

We establish the consistency of the proposed test with $\ell_1$ regularization and appropriate network structures. 
Let $$\pwrate_n^{(1)}(k,\tau,d)=\tau k^{\frac{1}{2}} (\log(dn))^{\frac{3}{2}}  \big\{ n^{-\frac{1}{k+2}} + n^{-\frac{1}{2}}(dk)^{\frac{1}{2}} + n^{-\frac{2}{k+2}}\log(n) +  n^{-1}\log(n)dk \big\}.$$
\begin{theorem}\label{thm:max_power_l1}
Suppose Assumptions \ref{assump:sample_size} and \ref{assump:tail} hold, and $\Sigma^{(1)}$ satisfies Assumption \ref{assump:pd}.  For any $(k, \tau) \in \mathcal C^{(1)}$, take the corresponding ReLU neural network to be
$\mathcal N(D, (k, 2k, 6(k+1)N, \dots, 6(k+1)N, 1), s, 4k^{1/2}\tau \xytrunc )$, 
where $D \asymp \log n$, $s \asymp N D$, $N \ge 2^k \vee (4k^{1/2}\tau\xytrunc+1)e^k$ and 
\[ N \asymp n^{\frac{k}{2+k}} (\log n)^{-\frac{3k}{2+k}}.\]
If there exists a $(k_\star, \tau_\star) \in \mathcal C^{(1)}$ such that 
\begin{equation}\label{eq:power_signal_l1}
	PW_{k_\star, \tau_\star}^{(1)}(\mu, \nu) = \sup_{U \in \sparsestiefelstar} W(U_{\#}\mu, U_{\#}\nu)	\ge c_a \pwrate_n^{(1)}(k_\star, \tau_\star, d)
\end{equation}
for some sufficiently large constant $c_a>0$, then we have $\prob(T_{n}^{(1)} > q_{1-\alpha}) \to 1$ as $n \to \infty$.
\end{theorem}

\begin{proof}
We provide an outline of the proof of Theorem \ref{thm:max_power_l1}, with full technical details deferred to the Supplementary Material. As discussed before in Section \ref{subsec:power}, Lemma \ref{lem:truncation_Fn} shows that, for a given pair $(k, \tau)$, it is enough to focus on the smaller class $\mathcal F_{n}^{(1)} = \{ f\in \mathcal F: f(x) = f(\tilde x), \text{ where } \tilde x = P_{[-{ \tau \xytrunc} , { \tau \xytrunc} ]^k}(x) \}$.
Define
\begin{equation*}\label{eq:SnfU}
	S_n(f, U) = \frac{1}{\nxtest}\sum_{i\in\IXtest} f(U^\t X_i) - \frac{1}{\nytest}\sum_{j\in \IYtest}f(U^\t Y_j)
\end{equation*}
$$\tilde S_n(f, U) = \frac{1}{\nxfit} \sum_{i\in\IXfit} f(U^\t X_i) - \frac{1}{\nyfit} \sum_{j\in\IYfit} f(U^\t Y_j),$$
and the oracle solution pair 
$$(f_{0, \star}^{(1)}, U_{0, \star}^{(1)}) \in \mathop{\arg\max}_{f \in \mathcal F_{n}^{(1)}, U \in \sparsestiefelstar} \{ \expect f(U^\t X) - \expect f(U^\t Y) \}.$$

For simplicity, we denote $\hat U = \hat U_{k_\star, \tau_\star}^{(1)}$ and $\hat f = \hat f_{k_\star, \tau_\star}^{(1)}$.
Let 
\begin{equation*}\label{eq:breve-f}
	\breve f \in \mathop{\arg\max}_{f \in \mathcal F_{n}^{(1)}}  \bigg\{ \frac{1}{\nxfit}\sum_{i \in\IXfit} f(\hat U^\t X_i) - \frac{1}{\nyfit}\sum_{j\in\IYfit} f(\hat U^\t Y_j)  \bigg\}
\end{equation*}
and
$$\tilde f\in \mathop{\arg\min}_{f \in \mathcal N(D, \bp, s, B)} \|f - \breve f\|_\infty. $$

The key step is to establish that the quantity $S_n(\hat f, \hat U)$ provides an accurate approximation of $PW_{k_\star, \tau_\star}^{(1)}(\mu, \nu)$, as summarized below. This approximation forms the basis for establishing the consistency of the proposed test.
\begin{align*}
	S_n(\hat f, \hat U)
	\approx \mathbb{E}S_n(\hat f, \hat U)
	& \approx \mathbb{E}\tilde S_n(\hat f, \hat U)
	&& \text{(uniform law of large numbers)} \\
	& \ge \mathbb{E}\tilde S_n(\tilde f, \hat U)
	& & \text{(optimization of \eqref{eq:hatf})} \\
	& \ge \mathbb{E}\tilde S_n(\breve f, \hat U)
	& & \text{(NN approximation)} \\
	&\approx \mathbb{E}\tilde S_n(f_{0, \star}^{(1)}, U_{0, \star}^{(1)})
	&& \text{(uniform law of large numbers)} \\
	&\approx PW_{k_\star, \tau_\star}^{(1)}(\mu,\nu)
	&& \text{(truncation)}.
\end{align*}
\end{proof}

\begin{remark}\label{rmk:pwd_signal_l1}
If the corresponding optimal projection direction satisfies $ \|U_0\|_1 \le \tau_\star$, then $PW_{k_\star}(\mu, \nu) = PW_{k_\star, \tau_\star}^{(1)}(\mu, \nu)$. 
Note that the proposed test does not require the true projection $U_0$ to be sparse. Rather, the key requirement is the existence of a suitable function class (e.g., $\sparsestiefelstar$) that is sufficiently rich to distinguish the null hypothesis from the alternative.
\end{remark}

As indicated by \eqref{eq:power_signal_l1}, a larger $\tau_\star$ leads to higher sample complexity and, consequently, requires a stronger signal to distinguish from noise and achieve favorable power. Introducing $\ell_1$ regularization enhances adaptivity by allowing the method to exploit potential sparsity in the discriminative projection directions.
In addition, Theorem~\ref{thm:max_power_l1} suggests that power can be further improved by incorporating multiple projection dimensions $k$, a benefit that is corroborated by the numerical results in the Supplementary Material. Therefore, the max-type test statistic in \eqref{eq:max_stat_l1} kills two birds with one stone: It not only eliminates the need of tuning $k$ and  $\tau$ {(though a candidate set $\mathcal C^{(1)}$ still needs to be specified)}, but also potentially increases the power via aggregating the strength from multiple projection dimensions and regularization parameters. 
In contrast to existing methods that require both candidate set selection and parameter tuning, our approach requires only the former, substantially reducing methodological and computational complexity.

\begin{remark}\label{rmk:candidate-set}
The proposed test is consistent as long as there exists a parameter pair $(k_\star,\tau_\star)$ in the candidate set that satisfies the signal condition, highlighting its adaptive nature. This feature distinguishes our method from approaches that rely on pre-specified parameters or data-driven tuning. Finally, numerical results in the Supplementary Material demonstrate that the proposed max-type test statistic achieves a significant power gain over tuning-based alternatives.
\end{remark}


\begin{remark} \label{rmk:diff_wang}
Our testing strategy differs significantly from the permutation test in \citet{wang2021two} in several key aspects. First, our method involves estimation of the discriminative projection directions and the discriminative function using deep neural networks. The use of deep neural networks to compute the maximum has the potential to amplify the signal. To see this, from the proof of Theorem \ref{thm:max_power_l1}, with probability tending to 1,
\[ \max_{j=1,\dots,m} |S_n(\hat g_j^{(1)})| \ge \sup_{U \in \mathcal S_{d, k, \tau}^{(1)}} W(U_{\#}\mu, U_{\#}\nu)- c \pwrate_n^{(1)}(k, \tau, d),   \]
for any $(k, \tau) \in \mathcal C$ and some constant $c>0$. Using deep neural networks enables better approximation of the maximizer in \eqref{eq:hatf}, resulting in a larger $\max_{j=1,\dots,m} |S_n(\hat g_j^{(1)})|$ and thus achieving higher power.
Second, we establish the limiting null distribution of the proposed test statistic and adopt its quantile to determine the test threshold. Third, the proposed approach is adaptive to potentially sparse structures in projection directions within high-dimensional contexts by incorporating a flexible sparsity parameter $\tau$. Fourth, our approach bypasses the need to select the projection dimension $k$ and the sparsity parameter $\tau$ via a max-type test statistic. Notably, the max-type test statistic yields superior power performance compared to the method depending on parameter tuning or a single projection dimension, as discussed in the Supplementary Material. 
\end{remark}

\subsection{Computational algorithm}\label{subsec:alg_l1}
The optimization problem \eqref{EQ:HATU_l1} involves projection onto the intersection of the Stiefel manifold $\stiefel$ and the $\ell_1$ ball, which is computationally nontrivial. To address this challenge, we instead consider the following penalized optimization problem,
\begin{align}\label{EQ:HATU_l1pen}
\max\limits_{U \in \stiefel} \min\limits_{\pi \in \Gamma(\mufit, \nufit)} \sum_{i\in\IXfit}\sum_{j\in \IYfit} \pi_{ij}\|U^{\t}(X_{i}-Y_{j})\|_2 - \rho\|U\|_1,
\end{align}
where $\rho \ge 0$.
A key observation is that the estimator $\hat U_\rho$ obtained from \eqref{EQ:HATU_l1pen} serves as the optimal solution to the constrained problem \eqref{EQ:HATU_l1} with $\tau = \|\hat U_\rho\|_1$. This property allows the selection of the constraint parameter $\tau$ to be transformed into the choice of an appropriate penalty parameter $\rho$ in \eqref{EQ:HATU_l1pen}.
Therefore, in practice, for a given collection of the candidate parameter pairs $(k_j, \rho_j), j=1, \dots, m$, we compute the corresponding estimates $\hat g_j^{(1)}$ and evaluate the final test statistic in an analogous manner.

In the above, the Stiefel manifold $\mathcal S_{d,k}$ constraint and the nonsmooth penalty function add to the complexity of the problem \eqref{EQ:HATU_l1pen}.  Existing algorithms \citep{lin2020projection,huang2021riemannian} are primarily designed for computing the entropic regularized projection Wasserstein distance, and the nonsmooth penalty renders these algorithms less applicable for our task. In Appendix \ref{apx:alg}, we provide an alternating optimization algorithm for \eqref{EQ:HATU_l1pen}, which leads to satisfactory numeric performances.

\section{Power enhancement with  $\ell_0$ regularization}\label{sec:l0}
Let the constraint set for the projection in \eqref{EQ:HATU} be $\zerosparsestiefel = \{ U \in \stiefel: \|U\|_0 \le \varpi\}$ where $k \le \varpi \le kd$.

\subsection{Estimation and the test statistic}\label{subsec:statistic_l0}

Under $\ell_0$ regularization, the optimization problem in \eqref{EQ:HATU} becomes
\begin{align}\label{EQ:HATU_l0}
\hat U_{k, \varpi}^{(0)} \in \mathop{\arg\max}\limits_{U \in \zerosparsestiefel} \bigg\{ \min\limits_{\pi \in \Gamma(\mufit, \nufit)} \sum_{i\in\IXfit}\sum_{j\in \IYfit} \pi_{ij}\|U^{\t}(X_{i}-Y_{j})\|_2 \bigg\}.
\end{align}
Given $\hat U_{k, \varpi}^{(0)}$, we then use the form in \eqref{eq:hatf} from Step 2 to obtain the estimator $\hat f_{k, \varpi}^{(0)}$, with the network parameters suggested in Thereom \ref{thm:max_power_l0}. 

Given a candidate set of hyperparameters $\mathcal C^{(0)} = \{ (k_j, \varpi_j), j=1, \dots, \mparam \}$, let $\hat S_{n}^{(0)} = (S_n(\hat g_{1}^{(0)}), \dots, S_n(\hat g_{\mparam}^{(0)}))^\t$, where $\hat g_{j}^{(0)}$ is the corresponding estimator using the hyperparameter pair $(k_j,\varpi_j)$ for $j=1,\dots, \mparam$.  
The corresponding max-type test statistic is given by
\begin{equation}\label{eq:max_stat_l0}
T_{\nm}^{(0)} =  \max_{j=1,\dots,\mparam}  \bigg(\frac{\nxtest\nytest}{\nxtest+\nytest}\bigg)^{1/2} |e_j^{\t} (\hat{\Sigma}^{(0)})^{-1/2} \hat S_{\nm}^{(0)}|,
\end{equation}
where $\hat \Sigma^{(0)}$ is the sample estimate of 
\begin{equation}\label{eq:Sigma_l0}
\Sigma^{(0)} = \var\{(\nxtest\nytest/(\nxtest+\nytest))^{1/2}\hat S_{\nm}^{(0)} \mid \data_{fit}\}. 
\end{equation}
Given a significance level $\alpha\in (0,1)$, we reject the null hypothesis if $T_{\nm}^{(0)} > q_{1-\alpha}$, where $q_{1-\alpha}$ is the $1-\alpha$ quantile of $\max_{j=1,\dots,\mparam} |Z_j|$ with $Z \sim N(0, I_\mparam)$.

\begin{corollary}\label{cor:size_l0}
Suppose that Assumption \ref{assump:sample_size} holds and $\Sigma^{(0)}$ satisfies Assumption \ref{assump:pd}. Under the null hypothesis, we have $T_\nm^{(0)} \distconverge \max_{j=1,\dots,\mparam} |Z_j|$, as $\nplusm \to \infty$, where $Z \sim N(0, I_\mparam)$. Moreover, with the network structures specified in Theorem \ref{thm:max_power_l1}, we have
\[  \sup_{t \in \real} \bigg|\prob(T_{\nm}^{(0)} \le t) - \prob\bigg(\max_{j=1,\dots, \mparam} |Z_j| \le t\bigg) \bigg| \le c\frac{ k_\circ^{3/2}\eta^3 (\varpi_\circ\log(dn))^{3/2} + k_\circ\eta^2\varpi_\circ\log(dn)\log(n)}{n^{1/2} }  \]
for some positive constant $c>0$, where $k_\circ = \max\{k: (k, \varpi)\in \mathcal C^{(0)}\}$ and $\varpi_\circ = \max\{\varpi: (k, \varpi)\in \mathcal C^{(0)}\}$.
\end{corollary}


The weak convergence of the proposed test statistic under $H_0$ guarantees asymptotic control of the type I error at level $\alpha$. 
Additionally, Corollary \ref{cor:size_l0} establishes the Berry--Essen type bound for the Gaussian approximation of the test statistic under $\ell_0$ regularization, with the convergence rate of order $n^{-1/2}\log(dn)\{ (\log(dn))^{1/2} + \log(n) \}$ if $\varpi_\circ \lesssim 1$ . 

\subsection{Power analysis}\label{subsec:power_l0}
We study the power performance under $\ell_0$ regularization. 
Define
\[ \pwrate_n^{(0)}(k, \varpi, d) = (\varpi k)^{\frac{1}{2}} (\log(dn))^{\frac{3}{2}}  \big\{ n^{-\frac{1}{k+2}} + n^{-\frac{1}{2}}(k\varpi)^{\frac{1}{2}} + n^{-\frac{2}{k+2}}\log(n) + n^{-1}\log(n)k\varpi \big\}.  \]
\begin{theorem}\label{thm:max_power_l0}
Suppose Assumptions \ref{assump:sample_size} and  \ref{assump:tail} hold, and $\Sigma^{(0)}$ satisfies Assumption \ref{assump:pd}.  For any $(k, \varpi) \in \mathcal C$, take the corresponding ReLU neural network to be
$\mathcal N(D, (k, 2k, 6(k+1)N, \dots, 6(k+1)N, 1), s, 4(k\varpi)^{1/2}\xytrunc)$, 
where $D \asymp \log n$, $s \asymp N D$, $N \ge 2^k \vee (4(k\varpi)^{1/2}\xytrunc+1)e^k$ and 
\[ N \asymp n^{\frac{k}{2+k}} (\log n)^{-\frac{3k}{2+k}}.\]
If there exists a $(k_\star, \varpi_\star) \in \mathcal C^{(0)}$ such that 
\begin{equation}\label{eq:power_signal_l0}
	PW_{k_\star, \varpi_\star}^{(0)}(\mu, \nu) = \sup_{U \in \zerosparsestiefelstar} W(U_{\#}\mu, U_{\#}\nu) \ge c_a \pwrate_n^{(0)}(k_\star, \varpi_\star, d)
\end{equation}
for some sufficiently large constant $c_a>0$, then we have $\prob(T_{n}^{(0)} > q_{1-\alpha}) \to 1$ as $n \to \infty$.
\end{theorem}

Compared to the signal conditions in Theorem \ref{thm:max_power_l1},  the condition in \eqref{eq:power_signal_l0} is considerably weaker provided that $\varpi_\star \lesssim 1$, reflecting the reduced complexity induced by the $\ell_0$ sparsity structure. 
Moreover, Theorem \ref{thm:max_power_l0} suggests that incorporating multiple projection dimensions and sparsity parameters can improve power.

\begin{remark}
If the corresponding optimal projection direction satisfies $ \|U_0\|_0 \le \varpi_\star$, then $PW_{k_\star}(\mu, \nu) = PW_{k_\star, \varpi_\star}^{(0)}(\mu, \nu) $. 
Again, the proposed test does not require the true projection $U_0$ to be sparse. Rather, the essential requirement is the existence of a suitably rich function class capable of distinguishing the null hypothesis from the alternative.
\end{remark}

\subsection{Computational algorithm}\label{subsec:alg_l0}
Solving the optimization problem \eqref{EQ:HATU_l0} is computationally challenging due to the nonconvex nature of the $\ell_0$ regularization. To address this difficulty, we adopt a sequential $\ell_1$ approximation strategy to find the solution within $\zerosparsestiefel$. Specifically, we consider a sequence of $\ell_1$ sparisity parameters $\rho_j, j=1, \dots, L$, and for each $\rho_j$ solve the following optimization problem
\begin{align}
\hat U_{\rho_j} \in \mathop{\arg\max}\limits_{U \in \stiefel} \bigg\{ \min\limits_{\pi \in \Gamma(\mufit, \nufit)} \sum_{i\in\IXfit}\sum_{j\in \IYfit} \pi_{ij}\|U^{\t}(X_{i}-Y_{j})\|_2 - \rho_j\|U\|_1 \bigg\}.
\end{align}
For each resulting candidate, we retain the $\varpi$ largest entries in absolute value and then apply a QR decomposition to obtain an orthonormal projection. Since the QR step may densify the projection, we repeat this pruning procedure several times. Among these candidates, we select the projection that attains the largest empirical projection Wasserstein distance while satisfying the constraint, and take it as the final estimator $\hat U_{k, \varpi}^{(0)}$. Numerical results demonstrate the satisfactory performance of the proposed algorithm.

\section{Simulation studies}
\label{sec:sim}

Let $0_d$ and $1_d$ be $d$-dimensional vectors with elements all equal to 0 and 1, respectively. Let $I_d$ denote the $d\times d$ identity matrix. We consider the following four models of different characteristics.

\begin{itemize}

\item Model A (mean shift): let $X \sim N(0_d, \Sigma)$ and $Y \sim N(\mu, \Sigma)$ with $\Sigma =(r_{ij})_{i,j=1}^d$, where $r_{ij} = 0.5^{|i-j|}$ and $\mu = (\mu_1, \dots, \mu_d)^{\t}$ with $\mu_j = 0.8\beta j^{-3}$.

\item Model B (variance shift): let $X \sim N(0_d, \Sigma)$ and $Y \sim N(0_d, \Sigma')$, where $\Sigma$ is defined in Model A and $\Sigma'$ is defined by replacing the diagonal entries of $\Sigma$ with $\Sigma'_{jj}=1 + 8\beta j^{-3}$ for $j=1,\dots,d$.

\item Model C (same mean and variance, different marginals):  let $X \sim 0.5 N(-1_d, I_d) + 0.5 N(1_d, I_d)$ and $Y=(Y_1^{\t}, Y_2^{\t})^{\t}$ with independent components $Y_1 \sim N(0_s, 2R)$ and $Y_2 \sim 0.5 N(-1_{d-s}, I_{d-s}) + 0.5 N(1_{d-s}, I_{d-s})$, where $R = (r_{ij})_{i,j=1}^s$, $r_{ij} = 0.5^{|i-j|}$. Set $s=[ 5\beta ]$, where $[\cdot]$ denotes rounding to the nearest integer.

\item  Model D (same univariate marginals, different joint distributions): let $X \sim 0.5 N(-1_d, I_d) + 0.5 N(1_d, I_d)$ and $Y \sim 0.5N(-1_d, \Sigma) + 0.5N(1_d, \Sigma)$, where 
$$\Sigma = \bigg(\begin{array}{cc}
	R & 0 \\
	0 & I_{d-s}
\end{array}\bigg), $$ $R = (r_{ij})_{i,j=1}^s$ and  $r_{ij}=0.9^{|i-j|}$. Set $s=[ 100\beta]$.

\end{itemize}

In all considered models, we take different values $\beta = 0, 0.2, 0.4, \dots, 1$, representing different signal levels. Specifically,  $\beta=0$ corresponds to the true null hypothesis, and  $\beta \neq 0$ signifies the presence of an alternative hypothesis. The larger the value of $\beta$, the stronger the signal becomes. Models A and B exhibit decaying patterns in mean and variance shifts, which are more realistic and challenging scenarios in high-dimensional contexts. Models C and D encompass distributional differences beyond mean and variance shifts in non-Gaussian settings. 

We compare the performance of the proposed method with several nonparametric approaches for the high-dimensional two-sample testing, including MMD \citep{gretton2012kernel}, MMD with a deep kernel (MMD-D) \citep{liu2020learning}, the energy distance ($\mathrm{ED}_2$) \citep{szekely2004}, the projected Wasserstein distance (PW) \citep{wang2021two}, the kernel Fisher discriminant analysis (KFDA) witness test \citep{kubler2022witness}.
Specifically, we use the Gaussian kernel with the bandwidth parameter determined by the median heuristic, i.e., the median disitance between points in the aggregated samples \citep{gretton2012kernel} to implement MMD.  
For MMD-D, we use the implementations provided in \citet{liu2020learning}.
In the case of PW, we set the projected dimension to 3, which is adopted in \cite{wang2021two}.
To determine the test thresholds, a permutation approach is employed for the MMD, MMD-D, PW and $\mathrm{ED}_2$ methods, whereas the asymptotic threshold is adopted for KFDA.
Note that in high-dimensional scenarios, the MMD with permutation and the studentized MMD with asymptotic normal distributions yield similar performance, as demonstrated in \citet{gao2023two}. Moreover, as illustrated in Figures S.1 and S.2 in Section S.1.1 of the Supplementary Material, the asymptotic normality of the studentized MMD is not trustworthy for low-dimensional settings. Therefore, we only consider the permutation-based MMD for numerical comparisons.

\begin{figure*}[t]
\begin{center}
	\newcommand{\thiswidth}{0.4\linewidth}
	\newcommand{\thisgap}{1mm}
	\begin{tabular}{cc}
		\hspace{\thisgap}\includegraphics[width=\thiswidth]{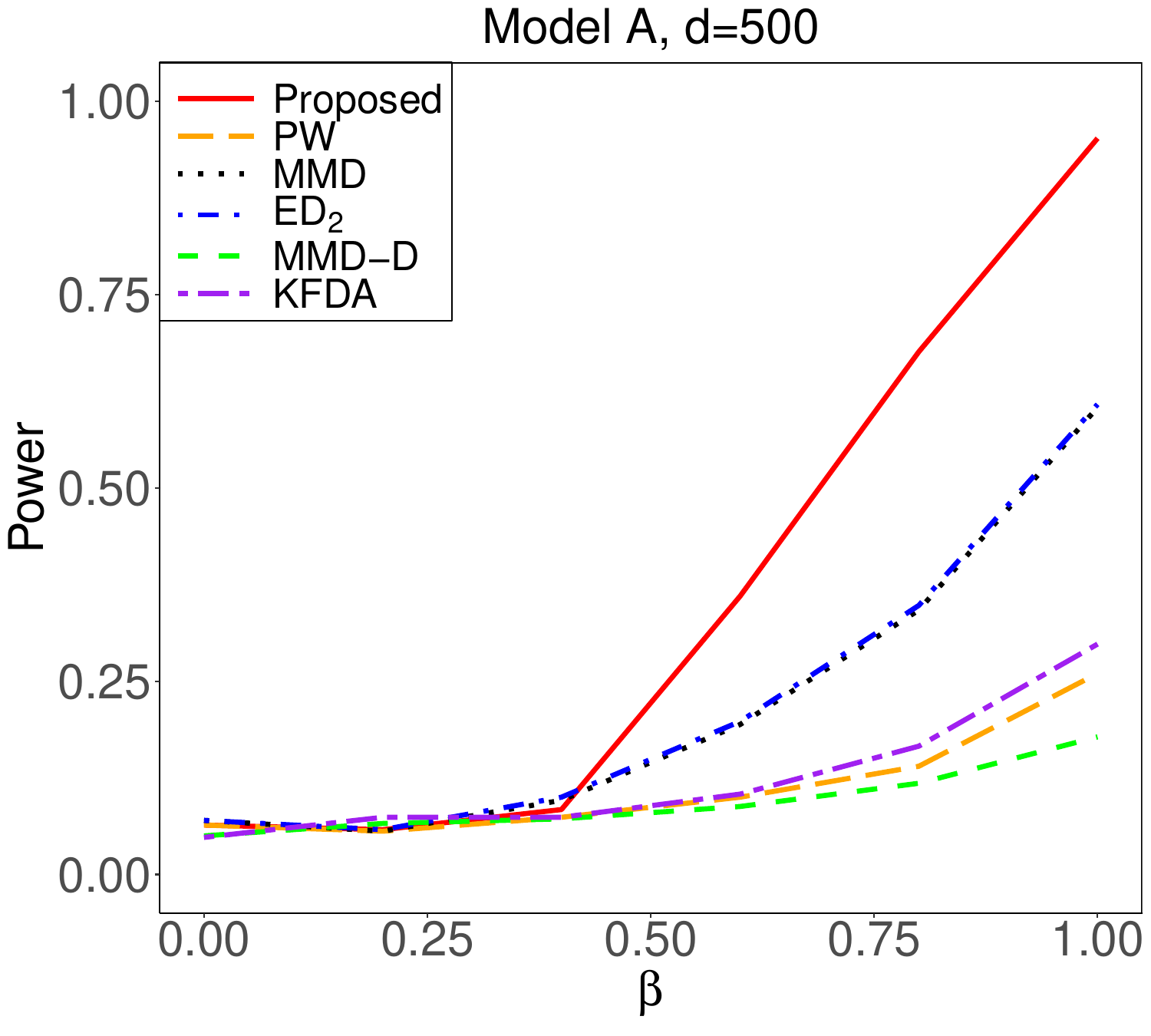} &
		\hspace{\thisgap}\includegraphics[width=\thiswidth]{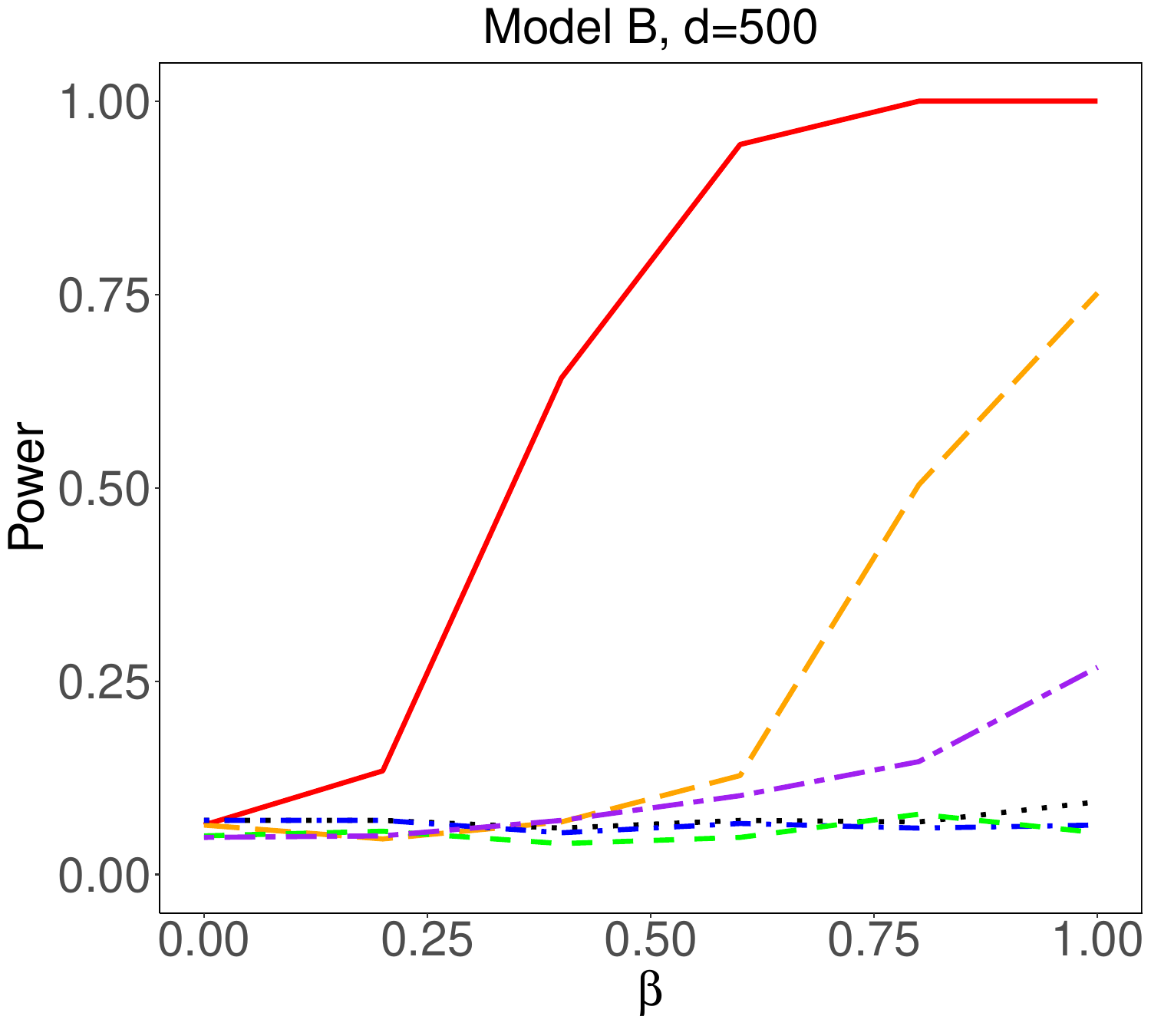} \\
		\hspace{\thisgap}\includegraphics[width=\thiswidth]{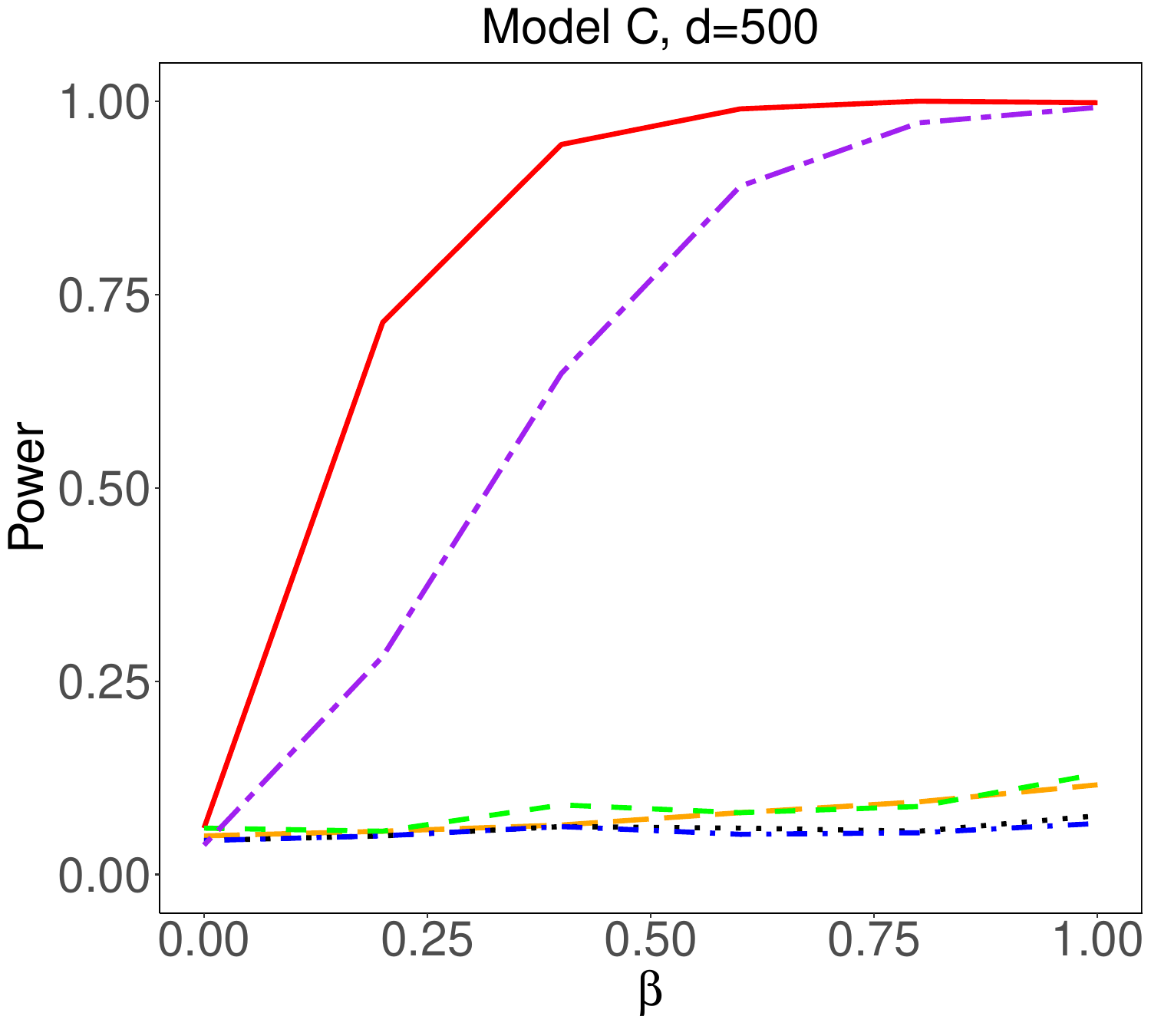} &
		\hspace{\thisgap}\includegraphics[width=\thiswidth]{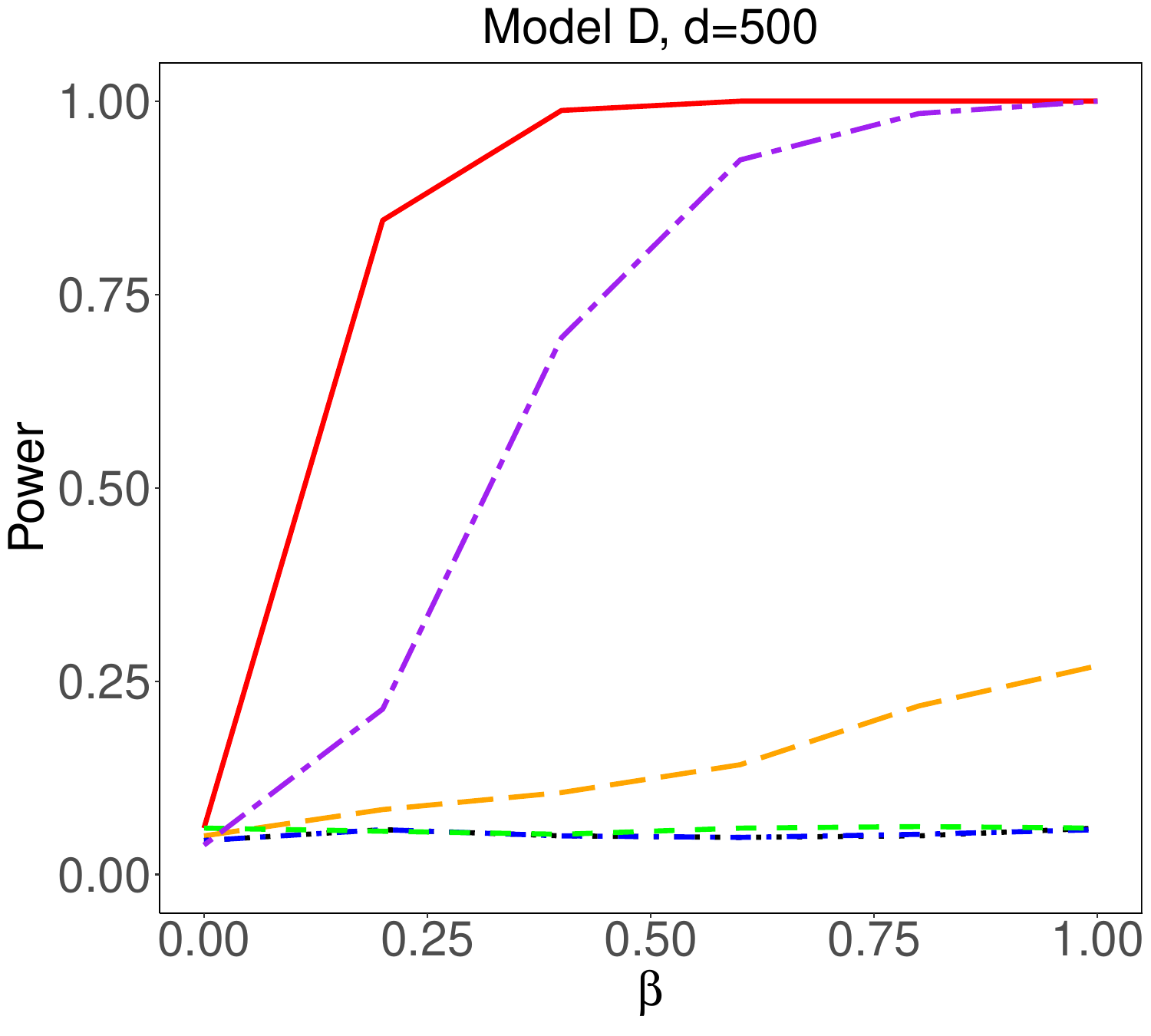}
	\end{tabular}
	\caption{Empirical size ($\beta=0$) and power $(\beta>0)$ of various methods in different models with $\nx=\ny=250$ and $d=500$.}
	\label{fig:sim_d500}
\end{center}
\end{figure*}

For the proposed method, we evenly split the data into $\data_{fit}$ for estimating $(\hat f,\hat U)$ and $\data_{test}$ for constructing the test statistic. The effect of the splitting ratio on the test performance is examined in Section S.1.5 of the Supplementary Material, where equal splitting yields the best overall performance. 

\subsection{Numerical results under $\ell_1$ regularization}\label{subsec:sim_l1}
{As discussed in Section \ref{subsec:alg_l1}, with  $\ell_1$ regularizaiton, we compute the projection directions via a penalized optimization approach and then construct the final test statistic accordingly.}
The candidate set consists of the combinations of $k=1, 5, 10$ and $\rho = 0.01, 0.1, 1$, covering regimes ranging from dense to sparse. According to Theorem \ref{thm:max_power}, under $n_x=n_y=250$, we use the ReLU neural network with 3 hidden layers, each layer containing an equal number of nodes. Specifically, the number of nodes per hidden layer is set to 50, 150, 250 for $k=1, 5, 10$, respectively.  To preserve the Lipschitz property, we practically adopt the spectral normalization technique \citep{miyato2018spectral}. In all settings, the significance level is $\alpha=0.05$, and we conduct 500 independent Monte Carlo repetitions to calculate the empirical percentage of rejections for each method. 

As depicted in Figure \ref{fig:sim_d500}, all the considered methods successfully control the type I error. In terms of power performance, in Model A where the mean shift occurs, the proposed method achieves the highest power when the signal is suitably large.
The MMD and $\mathrm{ED}_2$ yield comparable power, and both of them outperform the KFDA, PW and MMD-D. In Model B with distributional differences in variances, the proposed method yields the highest power. {The PW method shows reasonably increasing power as $\beta$ increases, while the power of the method KFDA displays a gradual increase. By comparison, the MMD, $\mathrm{ED}_2$ and MMD-D methods have limited power. 

Models C and D represent more challenging scenarios, where the methods MMD, MMD-D and $\mathrm{ED}_2$ exhibit restricted power in detecting the distributional differences. The PW displays little power in Model C, while its power increases slowly in Model D. KFDA exhibits excellent power in both models,  while the proposed method, enjoying the highest power, outperforms KFDA by a significant margin in these scenarios. 

To examine the effect of the increasing dimensionality on the power while keeping the signal fixed \citep{ramdas2015decreasing}, we consider several dimensions $d=60, 200, 500$. For the signal level, we set $\mu_j=0.64j^{-3}$ for Model A and $\Sigma'_{jj}=1+6.4j^{-3}$ for Model B. Additionally, we assign values $s=5$ and $s=60$ for Models C and D, respectively. As shown in Table \ref{tab:sim}, the proposed method consistently outperforms the other approaches in all Models. Furthermore, among all the considered approaches, the performance of the proposed test does not deteriorate much as the dimension increases. Notably, the KFDA delivers promising results in Models C and D, but its power decays rapidly in Models A and B. In contrast, the other methods exhibit a sharp decay of the power across all four models. 

\begin{table}[!t]
	\caption{The empirical power under different dimensions with $\nx=\ny=250$. \label{tab:sim}}
	\begin{tabular*}{\columnwidth}{@{\extracolsep{\fill}}cccccccc@{\extracolsep{\fill}}}
		\hline
		\multicolumn{2}{c}{} & Proposed & MMD & $\mathrm{ED}_{2}$ & MMD-D & KFDA & PW \\ \hline
		\multirow{3}{*}{Model A} & $d=60$ & 0.910 & 0.912  & 0.910 & 0.380 & 0.742 &  0.692 \\ 
		& $d=200$ & 0.832  & 0.608 & 0.614 & 0.192 & 0.330 &  0.378 \\ 
		& $d=500$ & 0.676   & 0.342 & 0.348 & 0.118 & 0.166 &  0.140  \\ \hline
		\multirow{3}{*}{Model B} & $d=60$ & 0.998  & 0.958  & 0.508 & 0.878 & 0.984  & 1.000  \\ 
		& $d=200$ & 0.996 &  0.140 & 0.072 & 0.066 & 0.498 & 0.958  \\ 
		& $d=500$ & 1.000  &  0.068 & 0.060 & 0.078 & 0.146  & 0.504  \\ \hline
		\multirow{3}{*}{Model C} & $d=60$ & 1.000  & 0.996  & 0.358 & 0.998 & 1.000  &  1.000  \\ 
		& $d=200$ & 1.000 & 0.144  & 0.086 & 0.566 & 1.000 &  0.293 \\ 
		& $d=500$ & 0.998  & 0.076  & 0.066 & 0.088 & 0.972 & 0.116  \\ \hline
		\multirow{3}{*}{Model D} & $d=60$ & 1.000   & 0.940  & 0.326 & 1.000 & 1.000 & 1.000  \\ 
		& $d=200$ & 1.000  & 0.052  & 0.050 & 0.192 & 1.000 & 0.994  \\ 
		& $d=500$ & 1.000  & 0.048  & 0.048 & 0.062 & 0.984 &  0.142 \\ 
		\hline
	\end{tabular*}
\end{table}

The test results of the proposed method may exhibit variability due to data splitting. To mitigate this variability, we aggregate p-values across multiple random splits \citep{cai2024test, dai2024significance}. In Section S.1.6, we evaluate several popular p-value aggregation methods. Among them, the Cauchy combination test \citep{liu2020cauchy}, Bonferroni correction, Bonferroni-Geometric \citep{vovk2020combining} and Bonferroni-Cauchy demonstrate superior performance. Compared to single-split results, proper aggregation of p-values can significantly enhance test power. As the number of random splits increases, power generally improves, while the gain becomes marginal beyond 5 splits. Using 5 to 10 random splits achieves a good balance between test performance and computational efficiency. Additional details are provided in Section S.1.6 of the Supplementary Material.

To sum up, the numerical comparisons illustrate the strength of the proposed method in detecting general distributional discrepancies. Moreover, its performance appears robust to the increasing dimensionality. 

\subsection{Numerical results under $\ell_0$ regularization} \label{subsec:sim_l0}

In this section, we present numerical results under $\ell_0$ regularization using the algorithm in Section \ref{subsec:alg_l0}. The candidate set $\mathcal C^{(0)}$ includes all combinations of $k=1, 5, 10$ and $\varpi = k, [kd^{1/2}], kd$, covering regimes from sparse to dense, where $[kd^{1/2}]$ denotes the integer part of $kd^{1/2}$.

\begin{table}[!t]
	\caption{The empirical percentage of rejections under different models with $\nx=\ny=250, d=500$. \label{tab:sim_l0}}
	\begin{tabular*}{\columnwidth}{@{\extracolsep{\fill}}ccccccc@{\extracolsep{\fill}}}
		\hline
		& Proposed & MMD & $\mathrm{ED}_{2}$ & MMD-D & KFDA & PW \\ \hline
		Gaussian, $H_0$  & 0.058 & 0.070  & 0.070 & 0.050 & 0.048 &  0.064 \\ \hline
		Non-Gaussian, $H_0$  & 0.070 & 0.044  & 0.044 & 0.060 & 0.038 &  0.050 \\ \hline
		A, $\beta=0.6$  & 0.586 & 0.194  & 0.198 & 0.088 & 0.104 &  0.100 \\ 
		A, $\beta=0.8$ & 0.934   & 0.342 & 0.348 & 0.118 & 0.166 &  0.140  \\ \hline
		B, $\beta=0.4$ & 0.886  & 0.060  & 0.054 & 0.040 & 0.070  & 0.068  \\
		B, $\beta=0.6$ & 0.994  &  0.070 & 0.066 & 0.048 & 0.102  & 0.128  \\ \hline
		C, $\beta=0.2$ & 0.830  & 0.050  & 0.050 & 0.056 & 0.282 & 0.056  \\ \hline
		D, $\beta=0.2$ & 0.932  & 0.058  & 0.058 & 0.056 & 0.214 &  0.084 \\ 
		\hline
	\end{tabular*}
\end{table}

As shown in Table \ref{tab:sim_l0}, the proposed test method under $\ell_0$ regularization substantially outperforms competing methods in terms of power, while successfully controlling the type I error.

\section{Real data}\label{sec:real-data}

Among brain tumors that originate from glial cells, Gliomas are the most common type, which  can be subdivided into grade II, grade III, and grade IV (glioblastoma multiforme, GBM), based on the level of aggressiveness and histological characteristics \citep{louis20072007,hsu2019identification}. 
To accelerate the comprehensive understanding of the genetics of cancer, the Cancer Genome Atlas (TCGA) project was launched by the National Institute of Health (NIH), providing publicly available cancer genomics datasets \citep{tomczak2015review} which can be accessed through the \texttt{R} package \texttt{cgdsr}. 
In the TCGA project, grade II and III are classified as lower grade glioma (LGG). 
DNA methylation is a  DNA-based alteration commonly occuring in cancer, which can provide critical information about the presence, development and prognosis of cancer \citep{kulis2010dna,mikeska2014dna}. 
Therefore, it is of great interest to investigate DNA methylation patterns in the GBM and LGG groups. 

To study the DNA methylation patterns, we use the DNA methylation data \citep{cerami2012cbio} from the lower grade glioma and glioblastoma multiforme studies in the TCGA project, focusing on candidate prognostic genes related to the brain cancer in the Human Pathology Atlas \citep{uhlen2017pathology}.   After excluding the missing values, the dataset consists of $\nx=511$ LGG and $\ny=150$ GBM tumor samples with $p=207$ genes.

The p-value of the proposed test under both $\ell_1$ and $\ell_0$ regularization is smaller than $0.001$, indicating a significant difference between these two groups in DNA methylation. The result is consistent with the conclusion of existing studies linking the heterogeneity of DNA methylation to the progression of LGG to GBM \citep{mazor2015dna,klughammer2018dna}. 
Moreover, we display the differences in mean and covariance in Figures \ref{fig:meandiff} and \ref{fig:heatmap}. As illustrated in Figure \ref{fig:meandiff}, the mean values differ significantly between the GBM and LGG groups across various gene variables, with the mean values of most coordinates in GBM being lower than those in LGG. In terms of covariance matrices, we present the heatmaps of covariance matrices for the subset of variables with the 50 largest univariate 1-Wasserstein distances for better visualization. As depicted in Figure \ref{fig:heatmap}, the covariance matrix of GBM appears sparser than that of LGG. The investigation of the mean and covariance reveals a clear  difference in DNA methylation patterns between GBM and LGG tumor samples.

\begin{figure}[!t]
	\centering
	\includegraphics[width=0.5\linewidth, height=0.4\linewidth]{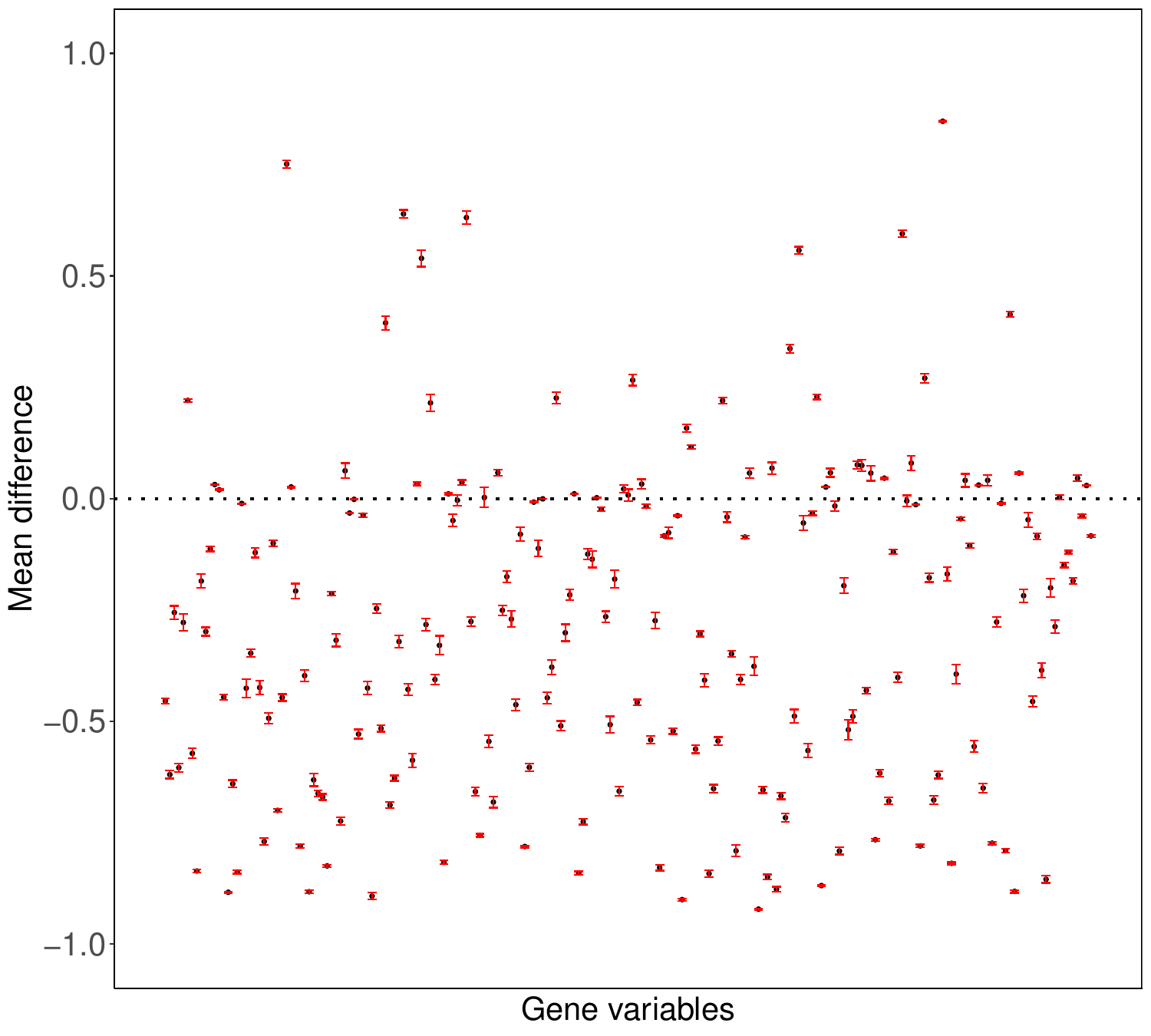} 
	\caption{The differences in mean between GBM and LGG, with error bars indicating the standard errors.}
	\label{fig:meandiff}
\end{figure}

\begin{figure}[!t]
	\centering
	\includegraphics[width=0.9\linewidth, height=0.4\linewidth]{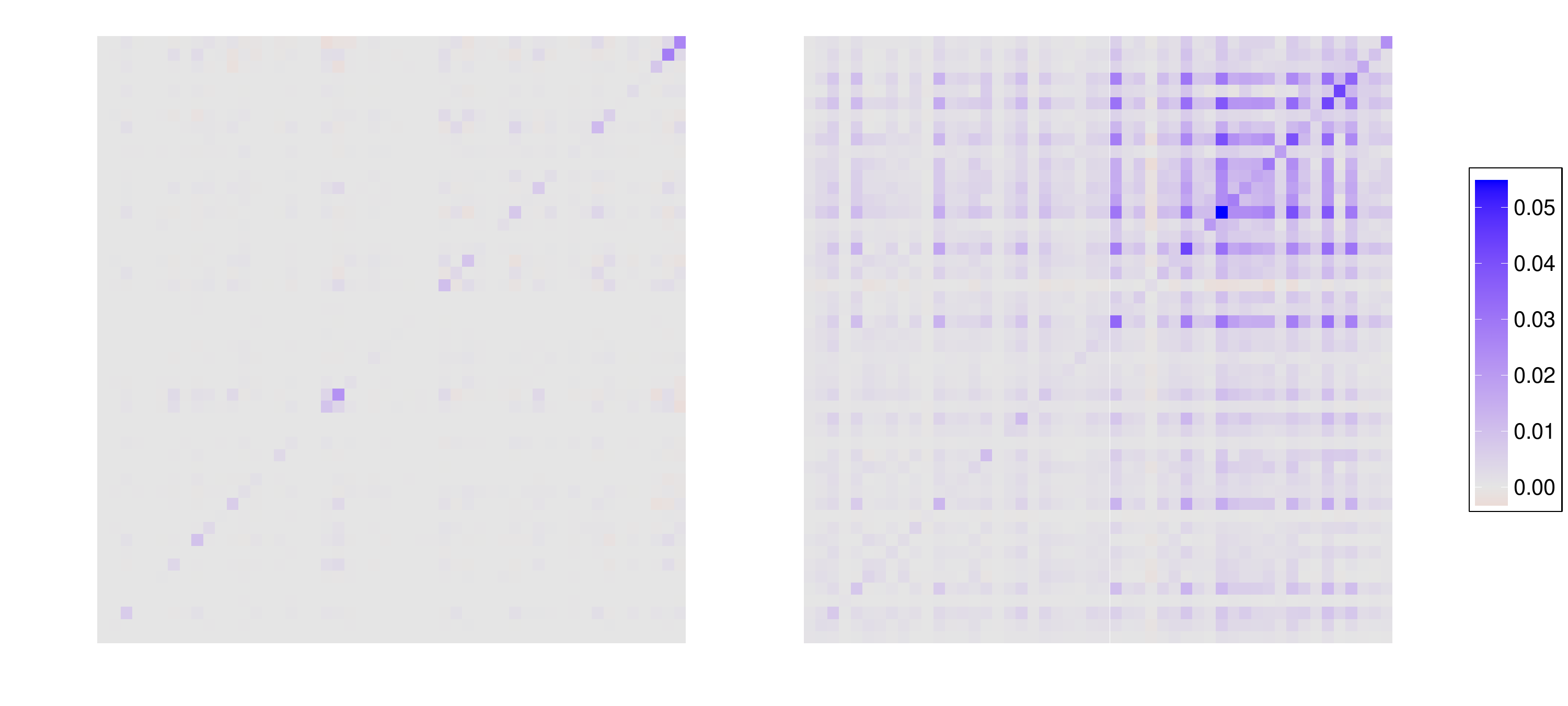} 
	\caption{The heatmaps of covariance matrices for the subset of variables with the 50 largest univariate 1-Wasserstein distances. Left: GBM; Right: LGG.}
	\label{fig:heatmap}
\end{figure}

\section{Conclusion}
\label{sec:conclusion}

In this paper, we employ deep neural networks and the projection  Wasserstein distance to develop a powerful test for high-dimensional two-sample testing problems. By aggregating strength from multiple values of hyperparameters via a max-type test statistic, the proposed method not only  avoids tuning crucial hyperparameters but also potentially increases the power. The test statistic is shown to asymptotically follow the distribution of the absolute maximum of a standard normal random vector, regardless of the estimation algorithms and the dimensionality. Such a pivotal statistic yields substantial computational efficiency when compared to data calibration methods such as permutation and bootstrapping. Despite the reduction in sample sizes due to data splitting, the proposed approach still attains excellent numerical performance, attributed to the utilization of multiple informative directions and advanced algorithms for learning discriminative functions via training deep neural networks.

While the sample splitting technique is handy for valid and effective inference in high-dimensional models \citep{rinaldo2019bootstrapping,cai2023asymptotic}, it is unclear whether such a loss of sample size is unavoidable.
Investigating this question is beyond the scope of this paper and left for future investigation.

\begin{appendix}
	\setcounter{equation}{0}
	\renewcommand{\theequation}{\Alph{section}.\arabic{equation}}

	\section{Optimization algorithm}\label{apx:alg}
	In this section, we provide detailed procedures for solving \eqref{EQ:HATU_l1pen}.
	Let $F(U; \pi) = l(U; \pi) + h(U)$ with $l(U; \pi) = -\sum_{i\in\IXfit}\sum_{j\in \IYfit} \pi_{ij}\|U^{\t}(X_{i}-Y_{j})\|_2$ and $h(U) = \rho \|U\|_1$.
	The optimal transport (OT) subproblem about $\pi$ is 
	\begin{equation}\label{op:ot}
		\max_{\pi \in \Gamma(\mu_{\nxfit}, \nu_{\nyfit})} l(U; \pi),
	\end{equation}
	which is solved by the linear programming, and the subproblem about $U$ is solved by the proximal gradient method for nonsmooth optimization over the Stiefel manifold (ManPG) \citep{chen2020proximal}.
	Specifically, the ManPG first computes a descent direction $V$ by solving the following problem
	\begin{equation} \label{eq:subop}
		\begin{split}
			\min_{V} & <\nabla l(U_t; \pi_{t+1}), V> + \frac{1}{2\gamma}\|V\|_F^2 + h(U_t+V) \\
			s.t. ~~& V^{\t}U_t + U_t^{\t}V = 0,
		\end{split}
	\end{equation}
	where $U_t$ is obtained in the $t$-th iteration, $\nabla l(U_t; \pi_{t+1})$ is the gradient of $l(\cdot ; \pi_{t+1})$, and $\gamma>0$ is a step size.
	Then, a retraction is performed to update $U$ on the Stiefel manifold. Algorithm \ref{alg:opt-U-pi} outlines this process. 
	
	Retraction operation is an important concept in manifold optimization; see \citet{absil2008optimization} for more details.
	There are many common retractions for the Stiefel manifold, including exponential mapping, the polar decomposition and the Cayley transformation. 
	For example, the exponential mapping \citep{edelman1998geometry} is given by,
	\[ \mathrm{Retr}_{U}(rV) = \left[\begin{array}{cc}
		U & Q 
	\end{array}\right]\mathrm{exp}\left(r\left[\begin{array}{cc}
		U^{\t}V & -R^{\t} \\
		R & 0
	\end{array}\right]\right)\left[\begin{array}{c}
		I_k \\
		0
	\end{array}\right], \]
	where $QR = (I_d - UU^{\t})V$ is the unique QR factorization.
	
	It remains to solve the subproblem \eqref{eq:subop}.
	Based on the Lagrangian function and the KKT system, we have
	\begin{equation}\label{eq:subop-lagrangian}
		\mathcal{A}_t(V(\Lambda)) = 0 ,
	\end{equation}
	and $V(\Lambda) = \mathrm{prox}_{\gamma h}(B(\Lambda)) - U_t$, 
	where $\mathcal{A}_t(V) = V^{\t}U_t + U_t^{\t}V$, $B(\Lambda) = U_t - \gamma(\nabla l(U_t; \pi_{t+1}) - 2U_t\Lambda)$, and $\Lambda$ is a $k \times k$ symmetric matrix. {For more details about the algorithm and the choice of hyperparameters}, readers can refer to \citet{chen2020proximal} and their public package. The semi-smooth Newton method \citep{xiao2018regularized} is used to solve \eqref{eq:subop-lagrangian} {about $\Lambda$.}
	
	\begin{algorithm}[t]
		\caption{The alternating optimization algorithm for solving \eqref{EQ:HATU_l1pen}.}
		\label{alg:opt-U-pi}
		\begin{algorithmic}[1]
			\Require $\{X_i: i\in\IXfit\}$, $\{Y_j: j \in \IYfit\}$, $U_0 \in \stiefel$, $\rho$, $t_{max}$, $\varepsilon$, $\gamma$, $\epsilon, \delta \in (0,1)$. 
			\For{$t=0$ {\bfseries to} $t_{max}$}
			\State Update $\pi_{t+1}$ given $U_t$ by solving the OT subproblem \eqref{op:ot}.
			\State Obtain $V_{t+1}$ by solving the subproblem (\ref{eq:subop}) given $\pi_{t+1}$.
			\State Set $r = 1$.
			\While {$F(\mathrm{Retr}_{U_t}(r V_{t+1}); \pi_{t+1}) > F(U_t; \pi_{t+1}) - \delta r \|V_{t+1}\|_F^2$}
			\State $r = \epsilon r$.
			\EndWhile
			\State Set $U_{t+1} = \mathrm{Retr}_{U_t}(r V_{t+1})$.
			\State Stop the iteration when $\|U_{t+1} - U_t\|_F/(1+\|U_t\|_F) \le \varepsilon$.
			\EndFor
		\end{algorithmic}
	\end{algorithm}
	
	\section{Remarks on $f_0$}\label{apx:f0}
	
	In Section \ref{sec:statistic}, we introduced the oracle discriminative function $f_0$. While this function does not directly enter our test procedure---it is introduced primarily to convey the intuition behind our method---we discuss below how $f_0$ can be conceptually chosen.
	
	In Section \ref{subsec:power}, we have showed that, for the purpose of empirical approximation of  $PW_k(\mu,\nu)$, it suffices to restrict attention to a smaller function class \( \mathcal{F}_n = \{ f \in \mathcal{F} : f(x) = f(\tilde{x}), \text{ where } \tilde{x} = P_{[-d^{1/2}\theta_n,  d^{1/2}\theta_n]^k}(x) \} \) that grows with $n$. 
	Given the compactness of \( \mathcal{F}_n \), the supremum
	\begin{equation}\label{eq:f0}
		f_{0} \in \underset{f \in \mathcal{F}_n}{\arg\max} \left\{ \mathbb{E} f(U_0^\t X) - \mathbb{E} f(U_0^\t Y) \right\}
	\end{equation}
	admits a solution. While this \( f_0 \) technically depends on \( n \), it is deterministic (non-random) and serves as the conceptual target that the empirical discriminative function $\hat f$ approximates.
	Note that the function $f_0$ in the above belongs to the 1-Lipschitz function class $\mathcal F$ since $\mathcal F_n\subset \mathcal F$.
	
	We also note that the function $f_0$ may not be unique. However, this does not affect the validity of our test or the theoretical guarantees. The key objective is to approximate the quantity $PW_{k}(\mu, \nu)$, rather than to recover a specific function $f_0$. In this sense, the essential step is the use of a neural network function class $\mathcal{N}$ (which grows with $n$) to approximate the class $\mathcal{F}_n$ (although $PW_k(\mu,\nu)$ depends on $\mathcal F$,  but $\mathcal F_n$ suffices here, as discussed in Section \ref{subsec:power}), and thereby approximate the quantity
	$$
	\sup_{f \in \mathcal{F}_n} \left\{ \mathbb{E} f(U_0^\t X) - \mathbb{E} f(U_0^\t Y) \right\}
	$$
	by $$
	\sup_{f \in \mathcal{N}} \left\{ \mathbb{E} f(U_0^\t X) - \mathbb{E} f(U_0^\t Y) \right\}.
	$$
	Thus, it is ultimately the expressive power of $\mathcal{N}$ that matters, rather than the approximation of a specific function $f_0$ within $\mathcal F_0$ or $\mathcal{F}$. 
	In this regard, Theorem \ref{thm:max_power} shows that the key statistic $S_n(\hat{f}, \hat{U})$,  constructed using the empirical function $\hat{f}$ which is based on the neural network, provides an accurate approximation of $PW_{k}(\mu, \nu)$.	

\end{appendix}

\section*{Acknowledgements}
	This research is partially supported by the NUS startup grant A-0004816-01-00.

\section*{Supplementary Material}
	The supplementary material contains additional simulation studies and technical proofs of the theoretical results. A \texttt{python} implementation of the proposed test can be found in \url{https://github.com/hxyXiaoyuHu/PWD-DNN}.

\bibliographystyle{asa} 
\bibliography{PWD}       

\end{document}